\documentclass[11pt,a4paper]{article}

\usepackage{fancyhdr}
\usepackage{amsmath,mathtools}
\usepackage{bm}
\usepackage{esint}
\usepackage{amsfonts}
\usepackage{amsthm}
\usepackage{amssymb}
\usepackage{amsbsy}
\usepackage{verbatim}
\usepackage{graphicx}
\usepackage{url}
\usepackage{mathrsfs}
\usepackage[hmargin = 1in,vmargin=.75in]{geometry}
\usepackage{color}
\usepackage{amstext}
\usepackage{stmaryrd}
\usepackage{enumerate}
\usepackage{algpseudocode}
\usepackage{algorithm}
\usepackage{pifont}
\usepackage{prettyref}
\usepackage{subcaption}
\font\msbm=msbm10

\numberwithin{equation}{section}

\theoremstyle{plain}
\newtheorem{theorem}{Theorem}[section]
\newtheorem{lemma}[theorem]{Lemma}

\newtheorem{definition}{Definition}[section]

\def\mathbb#1{\hbox{\msbm{#1}}}

\newcommand{\bu}{\boldsymbol{u}}

\newcommand{\bx}{\boldsymbol{x}}

\newcommand{\BA}{\boldsymbol{A}}

\newcommand{\BE}{\boldsymbol{E}}

\newcommand{\BG}{\boldsymbol{G}}

\newcommand{\BM}{\boldsymbol{M}}

\newcommand{\BO}{\boldsymbol{O}}
\newcommand{\BQ}{\boldsymbol{Q}}
\newcommand{\BR}{\boldsymbol{R}}
\newcommand{\BS}{\boldsymbol{S}}

\newcommand{\BU}{\boldsymbol{U}}
\newcommand{\BV}{\boldsymbol{V}}
\newcommand{\BW}{\boldsymbol{W}}
\newcommand{\BX}{\boldsymbol{X}}
\newcommand{\BY}{\boldsymbol{Y}}
\newcommand{\BZ}{\boldsymbol{Z}}

\newcommand{\BPhi}{\boldsymbol{\Phi}}
\newcommand{\BDelta}{\boldsymbol{\Delta}}

\newcommand{\BPsi}{\boldsymbol{\Psi}}

\newcommand{\BLambda}{\boldsymbol{\Lambda}}
\newcommand{\BSigma}{\boldsymbol{\Sigma}}

\newcommand{\PP}{\mathcal{P}}

\newcommand{\I}{\boldsymbol{I}}
\newcommand{\RR}{\mathbb{R}}
\newcommand{\lag}{\langle}
\newcommand{\rag}{\rangle}

\newcommand{\eps}{\epsilon}

\DeclareMathOperator{\Tr}{Tr}

\DeclareMathOperator{\mi}{\mathrm{i}}
\DeclareMathOperator{\E}{\mathbb{E}}

\DeclareMathOperator{\SO}{SO}
\DeclareMathOperator{\Od}{O}

\DeclareMathOperator{\spec}{spec}

\DeclareMathOperator{\rank}{\text{rank}}
\DeclareMathOperator{\blkdiag}{blkdiag}
\DeclareMathOperator{\argmin}{argmin}

\renewcommand{\Pr}{\mathbb{P}}
\long\def\\#1//{}

\definecolor{xl}{RGB}{200,50,120}

\begin{document}
\title{\bf Improved Performance Guarantees for \\ Orthogonal Group Synchronization \\ via Generalized Power Method}
\author{Shuyang Ling\thanks{New York University Shanghai (Email: sl3635@nyu.edu). This work is (partially) financially supported by the National Key R\&D Program of China, Project Number 2021YFA1002800, National Natural Science Foundation of China (NSFC) No.12001372, Shanghai Municipal Education Commission (SMEC) via Grant 0920000112, and NYU Shanghai Boost Fund.}}

\maketitle

\begin{abstract}
Given the noisy pairwise measurements among a set of unknown group elements, how to recover them efficiently and robustly? This problem, known as group synchronization, has drawn tremendous attention in the scientific community. In this work, we focus on orthogonal group synchronization that has found many applications, including computer vision, robotics, and cryo-electron microscopy. One commonly used approach is the least squares estimation that requires solving a highly nonconvex optimization program. The past few years have witnessed considerable advances in tackling this challenging problem by convex relaxation and efficient first-order methods. However, one fundamental theoretical question remains to be answered: how does the recovery performance depend on the noise strength? To answer this question, we study a benchmark model: recovering orthogonal group elements from their pairwise measurements corrupted by additive Gaussian noise. We investigate the performance of convex relaxation and the generalized power method (GPM). By applying the novel~\emph{leave-one-out} technique, we prove that the GPM with spectral initialization enjoys linear convergence to the global optima to the convex relaxation that also matches the maximum likelihood estimator. Our result achieves a near-optimal performance bound on the convergence of the GPM and improves the state-of-the-art theoretical guarantees on the tightness of convex relaxation by a large margin. 

\end{abstract}

\section{Introduction}
Group synchronization aims to recover group elements $\{g_i\}_{i=1}^n$ from their pairwise measurements:
\[
g_{ij} = g_ig_j^{-1} + \text{noise}, \quad (i,j)\in {\cal E}
\]
where ${\cal E}$ is a given edge set. It has attracted an increasing amount of attention in the past few years as it can be often found in various applications. Examples include computer vision (special orthogonal group SO($d$))~\cite{AKKSB12,CKS15}, robotics (special Euclidean group SE($d$))~\cite{RCBL19}, cryo-electron microscopy (SO(3))~\cite{SS12,S18,SS11}, clock-synchronization on networks~\cite{GK06}, sensor network localization (Euclidean group)~\cite{CLS12}, joint alignment ($\mathbb{Z}_n$-group)~\cite{CC18} and community detection ($\mathbb{Z}_2$-group)~\cite{ABBS14,C15}. Without noise, solving group synchronization is a trivial task since one can recover every group element sequentially and they are uniquely determined modulo a global group action. However, it is often highly challenging to solve the group synchronization problem in the presence of noise. When the noise exists, sequential recovery no longer works as it would amplify the noise. 

Recent few years have seen many theoretical and algorithmic progresses in solving group synchronization under various types of noise and with different underlying network structures~\cite{BBV16,LS21,LYS17,PWBM18b,RCBL19,ZB18}. 
In this work, we are interested in studying how to recover orthogonal group elements $\BG_i\in \Od(d)$ from the pairwise measurements corrupted by Gaussian noise, 
\begin{align}
\BA_{ij} & = \BG_i\BG_j^{\top} + \sigma\BW_{ij}, \label{eq:model1}\\ 
\Od(d) & : = \{\BO\in\RR^{d\times d}: \BO\BO^{\top} = \BO^{\top}\BO = \I_d\}, \nonumber %\label{eq:od}
\end{align}
where $\sigma$ measures the noise strength and $\BW_{ij} = \BW_{ji}^{\top}\in\RR^{d\times d}$ is a Gaussian random matrix with each entry an independent standard normal random variable. 

Why is~\eqref{eq:model1} an interesting model to study?
Orthogonal group is a generalization of many useful groups such as $\mathbb{Z}_2:=\{1,-1\}$ and unitary group $U(1) := \{e^{\mi \theta}:\theta\in[0,2\pi)\}$, and contains permutation groups and special orthogonal group as subgroups. Moreover, the $\Od(d)$ synchronization plays an important role in the special Euclidean group synchronization problem $\text{SE}(d)$, a core problem in SLAM (simultaneous localization and mapping)~\cite{RCBL19} and sensor network localization~\cite{CLS12}. Therefore, studying this benchmark model provides more insights into solving many relevant problems. 

Given this statistical model~\eqref{eq:model1}, one common approach is maximum likelihood estimation (MLE) which is equivalent to finding the least squares estimator. Minimizing the least squares cost function is a potentially NP-hard problem in general since it includes graph max-cut~\cite{GW95} as a special case $d=1$. On the other hand, many numerical simulations and previous theoretical works indicate that despite its seeming NP-hardness, the $\Od(d)$ synchronization is solvable with an efficient algorithm if noise level is relatively low. The focus of this paper is analyzing the performance of semidefinite relaxation and generalized power method in solving~\eqref{eq:model1}. More precisely, we attempt to answer two questions: 
\begin{center}
\text{\emph{When does convex relaxation recover the least squares estimator? }}
\end{center}
The first question concerns tightness: the solution to the relaxed convex program is not necessarily equal to the least squares estimator in general. When are they the same?
\begin{center}
\text{\emph{Can we design an efficient algorithm to solve this orthogonal group synchronization?}}
\end{center}
The second question is about efficient recovery: in practice, we do not often use the SDP (semidefinite programming) relaxation due to its expensive computational costs. As a result, finding an efficient algorithm with global convergence will be very favorable. In particular, we will study the convergence of generalized power method applied to the $\Od(d)$ synchronization~\cite{B16,LYS17,LYS20,ZB18}.
The answers to both questions are closely related to the noise strength, namely $\sigma$. As the noise strength increases, the recovery problem becomes more challenging. Therefore, we aim to identify a regime for $\sigma$ within which the generalized power method will successfully and efficiently recover the maximum likelihood estimator.

\subsection{Related works}

Since there is an extensive literature on group synchronization, we are not able to provide an exhaustive literature review. Instead, we will give a brief review of these works that inspire ours. 
Group synchronization has been studied for several different groups including $\mathbb{Z}_2=\{-1,1\}$~\cite{ABBS14,AFWZ20,BBV16,LXB19}, angular (phase) synchronization~\cite{BBS17,LYS17,S11,ZB18}, permutation group~\cite{CSG16,HG13,PKS13,SHSS16}, special Euclidean group~\cite{CLS12,RCBL19}, and finite cyclic group~\cite{CC18}. For the $\Od(d)$ synchronization, it has found broad applications in computer vision~\cite{AKKSB12,CGH14,OVBS17,S10}, generalized Procrustes problem~\cite{L21a,L21b,SS12}, and cryo-electron microscopy~\cite{S18,SS11}. 

Due to the practical importance of $\Od(d)$ synchronization, numerous efforts have been made towards developing fast and reliable algorithms, as well as mathematical theory.
One major difficulty of solving $\Od(d)$ synchronization is its severe nonconvexity in the least squares estimation. The nonconvexity makes naive local search algorithms vulnerable to poor local optima.
Therefore, one needs to come up with alternative approaches to overcome these issues. Convex relaxation, especially semidefinite relaxation, is viewed as one of the most powerful methods in solving these highly nonconvex programs. The SDP relaxation of $\mathbb{Z}_2$ synchronization is studied in~\cite{ABBS14} which obtains a near-optimal performance bound for the exact recovery of hidden group elements from corrupted measurements.
The work~\cite{S11} by Singer studies angular synchronization by spectral methods and semidefinite program in presence of random noise.  
Later on,~\cite{BBS17} provides the first yet suboptimal theoretical guarantee for the tightness of the SDP relaxation for angular synchronization. The near-optimal bound is first established in~\cite{ZB18} which uses the leave-one-out technique. In~\cite{WS13}, the authors studies the SDP relaxation approach for $\SO(d)$ synchronization with a subset of measurements corrupted by arbitrary random orthogonal matrices, and obtained a near-optimal performance guarantee with respect to the corruption level. The convex relaxation of $\Od(d)$ synchronization under arbitrary noise is recently studied in~\cite{L20a} which gives a sub-optimal performance bound. 

However, even if the SDP relaxation works in many scenarios, its expensive computational costs does not allow wide uses in practice. 
Instead, it is more practical to use fast low-rank nonconvex optimization approach for $\Od(d)$ synchronization including Riemannian optimization~\cite{AMS09,B15,BVB20,BM03,BM05,DRWMC20,MMMO17,RCBL19,WW20,WY13} since the solution is often low-rank. In particular, the Burer-Monteiro factorization~\cite{BM03,BM05} has gained substantial popularity due to its great empirical successes in solving large-scale SDPs. However, since the objective function associated with Burer-Monteiro factorization is nonconvex, there is always risk that the iterates might get stuck at one of the poor local optima. 
Recent theoretical progresses have shown that as long as the degree of freedom $p$ in the search space is large (scales like $p \approx \sqrt{2n}$ where $n$ is the number of constraints), the first and second order necessary optimality conditions are sufficient to guarantee global optimality~\cite{BVB20,WW20}, which provides a solid theoretical foundation for the Burer-Monteiro approach. Moreover, even if one chooses a search space of much smaller dimension, the Burer-Monteiro factorization still works~\cite{BVB16,L20b} in the examples such as community detection, and $\mathbb{Z}_2$- and $\Od(d)$ synchronization, provided that the noise is sufficiently small (not optimal). 
Other approaches of group synchronization includes the convenient spectral relaxation~\cite{L20b,RG19,S11,SS11}, often providing a solution close to the ground truth but not as good as the maximum likelihood estimator. Message passing type algorithms are also discussed in group synchronization~\cite{LS21,PWBM18b} in which a general algorithmic framework is established to solve a large class of compact group synchronization problems.

The performances of projected power method in joint alignment problem ($\mathbb{Z}_m$-cyclic group) and angular synchronization are studied in a series of works~\cite{B16,CC18,LYS17,ZB18}. The algorithm consists of an ordinary power iteration plus a projection to the group. One core question regarding the power method is: when can we recover the maximum likelihood estimator (MLE) via generalized power method? The answer depends on the noise level. The work~\cite{CC18} proves that projected power method with proper initialization converges to the MLE for joint alignment problem and achieves an information-theoretically optimal performance bound in the asymptotic regime.
Several papers~\cite{B16,LYS17,ZB18} have been devoted to solving the angular synchronization by using the generalized power method and to providing convergence analysis of the algorithm. The inspiring work~\cite{ZB18} employs the leave-one-out technique and is able to establish a near-optimal bound on the noise strength to ensure the global convergence of the generalized power method to the MLE. As a by-product,~\cite{ZB18} also provides a near-optimal bound for the tightness of the SDP relaxation. 
Recently,~\cite{LYS20} proposes an algorithmic framework based on the generalized power method to solve the synchronization of $\Od(d)$ group and its subgroups such as permutation group and special orthogonal group. The authors also provide a theoretical analysis on how far the output from the GPM is away from the ground truth. 

One important technical ingredient in our work is matrix perturbation theory~\cite{DK70,RCY11,V17,EBW18}, commonly used in statistics, signal processing, and machine learning. In particular, our analysis relies on matrix perturbation argument to analyze the spectral initialization.
However, the classical Davis-Kahan theorem~\cite{DK70} is insufficient to directly provide a tight bound because we need a block-wise operator norm bound for the eigenvectors while Davis-Kahan theorem only gives an error bound in $\ell_2$-norm. Our work benefits greatly from the leave-one-out technique in~\cite{AFWZ20,CCFMY20,FWZ18,MWCC20,ZB18} which is used to overcome this technical challenge. This technique has proven to be a surprisingly powerful method in dealing with the statistical dependence between the signal and  noisy measurements, and allowing us to achieve near-optimal performance bounds. We have seen many successful examples including $\mathbb{Z}_2$-synchronization and community detection under stochastic block model~\cite{AFWZ20}, statistical ranking problem~\cite{CFMW19}, $\ell_{\infty}$-norm eigenvector perturbation~\cite{FWZ18}, spectral clustering~\cite{DLS21}, and phase retrieval/blind deconvolution~\cite{LMCC19,MWCC20}, and matrix completion~\cite{CCFMY20,MWCC20}.

Our main contribution of this work is multifold. The work studies the global convergence of the generalized power method~\cite{LYS20} and the performance of the SDP relaxation~\cite{B15,RCBL19,BVB20,L20a} in solving orthogonal group synchronization with data corrupted by Gaussian noise. Despite the aforementioned progresses, there still exists a gap between the theory and practice. Numerical experiments indicate that the SDP relaxation and generalized power methods work if the noise strength $\sigma \lesssim \sqrt{n}/(\sqrt{d} + \sqrt{\log n})$ holds, modulo a possible log-factor. However, the state-of-the-art guarantee $\sigma\lesssim n^{1/4}/d^{3/4}$ in~\cite{L20a} is still far from this empirical observations in $\Od(d)$-synchronization for $d\geq 3.$
We establish a near-optimal performance bound for the global convergence to the maximum likelihood estimator as well as the tightness of the SDP relaxation. Our result narrows the gap between theory and empirical observation by showing that the GPM and the SDP relaxation work if $\sigma \lesssim \sqrt{n}/\sqrt{d}(\sqrt{d} + \sqrt{\log n})$, only losing a factor of $\sqrt{d}$ to the detection threshold~\cite{PWBM18}. Moreover, we provide a near-optimal error bound between the MLE and ground truth.
Our work can be viewed as a generalization of the projected power method for phase synchronization~\cite{ZB18} to orthogonal group synchronization~\cite{LYS20}. However, since U(1) in~\cite{ZB18} is a commutative group while orthogonal group $\Od(d)$ is non-commutative for $d\geq 3$, this fundamental difference requires many different technical treatments in the proof. Moreover, this result complements the work~\cite{LYS20} by answering what the limiting point of the GPM is, i.e., the convergence to the MLE for the iterates from generalized power method under~\eqref{eq:model1}.

\subsection{Organization}
We will discuss the preliminaries including model setup and algorithm in Section~\ref{s:prelim}. The main results will be provided in Section~\ref{s:main} and the proofs are given in Section~\ref{s:proof}.

\subsection{Notation}
We denote vectors and matrices by boldface letters $\bx$ and $\BX$ respectively. For a given matrix $\BX$, $\BX^{\top}$ is the transpose of $\BX$ and $\BX\succeq 0$ means $\BX$ is positive semidefinite. Let $\I_n$ be the identity matrix of size $n\times n$. For two matrices  $\BX$ and $\BY$ of the same size, their inner product is $\lag \BX,\BY\rag= \Tr(\BX^{\top}\BY) = \sum_{i,j}X_{ij}Y_{ij}.$ Let $\|\BX\|$ be the operator norm, $\|\BX\|_*$ be the nuclear norm, and $\|\BX\|_F$ be the Frobenius norm. {We denote the $i$th largest and the smallest singular value (and eigenvalue) of $\BX$ by $\sigma_{i}(\BX)$ and $\sigma_{\min}(\BX)$ (and $\lambda_i(\BX)$ and $\lambda_{\min}(\BX)$) respectively.}
For a non-negative function $f(x)$, we write $f(x)\lesssim g(x)$ and $f(x) = O(g(x))$ if there exists a positive constant $C_0$ such that $f(x)\leq C_0g(x)$ for all $x.$

\section{Preliminaries}\label{s:prelim}
Recall the orthogonal group synchronization under Gaussian noise:
\[
\BA_{ij} = \BG_i\BG_j^{\top} + \sigma\BW_{ij}\in \RR^{d\times d}
\]
where $\BG_i\in\Od(d)$ and $\BW_{ij}$ is a $d\times d$ Gaussian random matrix.
Let $\BG^{\top} = [\BG^{\top}_1, \cdots, \BG^{\top}_n]\in\RR^{d\times nd}$. Then we write the model into a more convenient block matrix form
\begin{equation}\label{eq:model}
\BA = \BG\BG^{\top} + \sigma \BW, \quad \BW\in\RR^{nd\times nd}
\end{equation}
where $\BW = [\BW_{ij}]_{1\leq i,j\leq n}$ is an $nd\times nd$ symmetric Gaussian random matrix whose entries are i.i.d. standard normal.
The task is to recover the group elements from their pairwise measurements $\BA_{ij}$. 

Among many existing approaches for $\Od(d)$ synchronization, one popular approach is maximum likelihood estimation (MLE). Under Gaussian noise, the likelihood function is equivalent to the least squares cost function:
\[
\min_{\BR_i\in \Od(d)}\sum_{i<j} \| \BR_i\BR_j^{\top} - \BA_{ij} \|_F^2
\]
which is equivalent to
\begin{equation}\label{def:od}
\max_{\BR_i\in\Od(d)} \sum_{i<j} \lag\BR_i\BR_j^{\top}, \BA_{ij}  \rag.
\end{equation}

Directly maximizing~\eqref{def:od} is not an easy task since~\eqref{def:od} is a highly nonconvex optimization program and potentially NP-hard. Convex relaxation is a powerful approach to overcome the nonconvexity. Let $\BR\in\Od(d)^{\otimes n}$ be a matrix of size $nd\times d$ whose $i$th block is $\BR_i.$ Note that every element in $\{\BR\BR^{\top}: \BR_i\in\Od(d)\}$ is positive semidefinite and its diagonal block equals $\I_d.$ The convex relaxation is obtained by ``lifting" the feasible set:
\begin{equation}\label{def:sdp}
\max~\lag \BA, \BX\rag \quad \text{ such that }\quad \BX\succeq 0,~ \BX_{ii} = \I_d.
\end{equation}
This semidefinite relaxation is a generalization of Goemans-Williamson relaxation~\cite{GW95} for graph max-cut. The global maximizer of~\eqref{def:sdp} can be recovered by using a standard convex program solver~\cite{GB14,NN94}. However, due to the relaxation, it is likely that the solution to~\eqref{def:sdp} does not match that to~\eqref{def:od}. Therefore, one major theoretical question is the tightness, i.e., when~\eqref{def:od} and~\eqref{def:sdp} share the same global maximizer. More generally, we are interested in when~\eqref{def:od} is solvable by an algorithm with polynomial time complexity. 
The answer depends on the value of $\sigma$: for~\eqref{def:od} with small $\sigma$, it is more likely to obtain the optimal solution. 

Our discussion focuses on answering two questions: (a) when is the SDP relaxation tight? How does the tightness depend on $\sigma$?; (b) if the SDP is tight, does there exist an efficient algorithm to recover the optimal solution to~\eqref{def:sdp} and~\eqref{def:od}? Now we introduce the generalized power method: the algorithm consists of a two-step procedure widely employed in nonconvex optimization literature~\cite{CLS15,MWCC20,ZB18}: first we find a good initialization by using spectral method and then show the global convergence of the generalized power method. Before proceeding to the algorithm, we give a few useful definitions.
{
\begin{definition}[Generalized ``phase"]
The generalized ``phase" of a matrix $\BPsi\in\RR^{d\times d}$ is defined as 
\begin{equation}\label{def:P}
\PP(\BPsi) : = \BU\BV^{\top} +\BU_{\perp}\BV_{\perp}^{\top}, \qquad \BPsi  = \BU\BSigma\BV^{\top}, 
\end{equation}
where $\BPsi = \BU\BSigma\BV^{\top}$ is the compact SVD of $\BPsi$, i.e., $\BU^{\top}\BU=\BV^{\top}\BV=\I_r$ and $r=\rank(\BPsi)$. Here $\BU_{\perp}$ and $\BV_{\perp}$ are $d\times (d-r)$ matrices satisfying $\BU_{\perp}^{\top}\BU=\BV_{\perp}^{\top}\BV=0$ and $\BU_{\perp}^{\top}\BU_{\perp}=\BV_{\perp}^{\top}\BV_{\perp}=\I_{d-r}$. 

For a block matrix $\BX^{\top} = [\BX_1^{\top},\cdots, \BX_n^{\top}]\in\RR^{d\times nd}$ with $\BX_i\in\RR^{d\times d}$, we define $\PP_n(\cdot)$ as an operator $\RR^{nd\times d}\rightarrow\Od(d)^{\otimes n}$:
\[
\PP_n(\BX) := 
\begin{bmatrix}
\PP(\BX_1) \\
\vdots \\
\PP(\BX_n)
\end{bmatrix}\in\Od(d)^{\otimes n}.
\]
\end{definition}
The operator $\PP(\cdot)$ is also known as the~\emph{matrix sign function}~\cite{AFWZ20,G11}.
Indeed, $\PP(\BPsi)$ is not unique because there are multiple choices of $\BU_{\perp}$ and $\BV_{\perp}$ if $\BPsi$ is not full rank. Therefore, it is more reasonable to treat $\PP(\cdot)$ as a set-valued function and 
 $\PP(\BPsi)$ outputs one representative from the set 
\[
\{\BU\BV^{\top} + \BU_{\perp}\BV_{\perp}^{\top}: \BU_{\perp}^{\top}\BU=\BV_{\perp}^{\top}\BV=0 \}, \qquad \BPsi = \BU\BSigma\BV^{\top},
\]
which is a subset of $\Od(d)$.
With a bit abuse of notation, a possibly simpler  alternative way of defining $\PP(\cdot)$ is to let $\BPsi =\BU\BSigma\BV^{\top}$ be any full SVD of $\BPsi$ and 
\[
\PP(\BPsi) := \BU\BV^{\top}\in\Od(d)
\]
where $\BU$ and $\BV$ are both $d\times d$ orthogonal matrices.}
In particular, if $\BPsi\in\RR^{d\times d}$ is invertible, then
\[
\PP(\BPsi) : = \BPsi(\BPsi^{\top}\BPsi)^{-\frac{1}{2}} = (\BPsi\BPsi^{\top})^{-\frac{1}{2}}\BPsi
\]
is uniquely determined.

The operator $\PP_n$ is essentially the projection operator which maps a matrix of $\RR^{nd\times d}$ to $\Od(d)^{\otimes n}$ since for each block $\BX_i$, $\PP(\BX_i)$ is the orthogonal matrix closest to $\BX_i$:
\[
\PP(\BX_i) : = \argmin_{\BQ\in\Od(d)} \|\BQ - \BX_i\|_F^2.
\]

In our discussion, we also need to introduce a new distance function between two matrices $\BX$ and $\BY$ in $\RR^{nd\times d}$. 
Remember that for orthogonal group synchronization, the solution is equivalent up to a global group action, i.e., for any fixed $\BG\in\Od(d)^{\otimes n},$ $\BG$ and $\BG\BQ$ are equivalent for any $\BQ\in\Od(d)$. This fact needs to be taken into consideration and we define the following distance function for two matrices in $\RR^{nd\times d}$: 
\begin{equation}\label{def:df}
d_F (\BY, \BX) := \min_{\BQ\in \Od(d)}\| \BY - \BX\BQ \|_F,
\end{equation}
where $\BQ = \PP(\BX^{\top}\BY)$ minimizes $\|\BY - \BX\BQ\|_F$ {since
\begin{align*}
d_F(\BY,\BX) & = \min_{\BQ\in\Od(d)}\sqrt{\|\BX\|_F^2 + \|\BY\|_F^2 - 2\lag \BX^{\top}\BY,\BQ\rag} \\
& = \sqrt{\|\BX\|_F^2 + \|\BY\|_F^2 - 2\| \BX^{\top}\BY\|_*}.
\end{align*}}
This distance function satisfies triangle inequality. For three arbitrary $nd\times d$ matrices $\BX,\BY,$ and $\BZ$, we have
\begin{align*}
d_F(\BX,\BZ) & = \| \BX - \BZ\BQ_{ZX}\|_F \leq \| \BX - \BZ\BQ_{ZY}\BQ_{YX} \|_F  \\
& \leq  \| \BX - \BY\BQ_{YX} \|_F + \| \BY\BQ_{YX} - \BZ \BQ_{ZY}\BQ_{YX} \|_F \\
& \leq d_F(\BX,\BY) + d_F(\BY, \BZ)
\end{align*}
where $\BQ_{ZX} = \PP(\BZ^{\top}\BX)$, $\BQ_{ZY} = \PP(\BZ^{\top}\BY)$, and $\BQ_{YX} = \PP(\BY^{\top}\BX)$ all belong to $\Od(d).$

Now we are ready to present the algorithm, which has been proposed in~\cite{LYS20,ZB18}.

\begin{algorithm}[h!]
\caption{Generalized power methods for orthogonal group synchronization}
\begin{algorithmic}[1]
\State Compute the top $d$ eigenvectors $\BPhi\in\RR^{nd\times d}$ of $\BA$ with $\BPhi^{\top}\BPhi = n\I_d.$
\State Compute $\PP(\BPhi_i)$ for all $1\leq i\leq n$ where $\BPhi_i$ is the $i$th block of $\BPhi$.
\State Initialize $\BS^{0} = \PP_n(\BPhi).$  

\State $\BS^{t+1} = \PP_n( \BA\BS^{t})$, \quad $t=0,1,\cdots$

\State Stop when the iteration stabilizes.

\end{algorithmic}
\label{algo1}
\end{algorithm}\label{alg}

The algorithm begins with a spectral initialization which first computes the top $d$ eigenvectors of $\BA$ and then projects them onto $\Od(d)^{\otimes n}.$ After that,
each step of the algorithm is essentially one step of the power iteration followed by projection. In order to better understand this algorithm, we first take a closer look at this algorithm and its possible fixed point. Let $\BS_i^t$ be the $i$th block of $\BS^t$ and the update rule for the $i$th block is
\[
\BS_i^{t+1} = \PP\left(\sum_{j=1}^n\BA_{ij} \BS_j^{t} \right) = {\BU_{i}^{t}}(\BV_i^{t})^{\top}, \quad 1\leq i\leq n,
\]
where $\BU_i^t$ and $\BV_i^t$ are the left and right singular vectors of $\sum_{j=1}^n\BA_{ij} \BS_j^{t}$.
Note that
\begin{equation}\label{eq:fixed}
\sum_{j=1}^n \BA_{ij}\BS_j^{t} = \BU_i^{t}\BSigma_i^{t}(\BV_i^{t})^{\top} = \BU_i^{t}\BSigma_i^{t}(\BU_i^{t})^{\top} \cdot \BU_i^{t}(\BV_i^{t})^{\top} = \BLambda_{ii}^{t}  \BS_i^{t+1} 
\end{equation}
where
$\BLambda_{ii}^{t} := \BU_i^{t}\BSigma_i^{t}(\BU_i^{t})^{\top}\succeq 0.$ 
{In other words, $\BLambda_{ii}^t = \left([\BA\BS^t]_i [\BA\BS^t]_i^{\top}\right)^{1/2}$ is the matrix square root of $[\BA\BS^t]_i [\BA\BS^t]_i^{\top}$ and is uniquely determined.}
A more compact form of the update rule is
\[
\BA \BS^t = \BLambda^t \BS^{t+1}
\]
where $\BLambda^t = \blkdiag(\BLambda_{11}^t, \cdots, \BLambda_{nn}^t)$ is a block-diagonal matrix whose diagonal blocks consist of $\{\BLambda_{ii}^t\}_{i=1}^n.$

Suppose $\{\BS^t\}_{t\geq 0}$ finally converges under $d_F(\cdot,\cdot)$, the limiting point $\BS^{\infty}$ is~\emph{likely} to satisfy (we will justify why it holds later):
\begin{equation}\label{eq:opt_eq}
\BA \BS^{\infty} = \BLambda \BS^{\infty}
\end{equation}
where $\BLambda = \blkdiag(\BLambda_{11}, \cdots, \BLambda_{nn})\in\RR^{nd\times nd}\succeq 0$. In fact,~\eqref{eq:opt_eq} arises again when we want to characterize the global optimality of $\BS^{\infty}(\BS^{\infty})^{\top}$ in~\eqref{def:sdp}: if $\BS^{\infty}$ satisfies~\eqref{eq:opt_eq} and $\BLambda - \BA\succeq 0$, then $\BS^{\infty}(\BS^{\infty})^{\top}$ is a global maximizer to~\eqref{def:sdp}, which follows from the duality theory in convex optimization.

\section{Main theorems}\label{s:main}
The first result improves the state-of-the-art bound for the tightness of the SDP relaxation to solve orthogonal group synchronization. 

\begin{theorem}\label{thm:sdp}
Consider orthogonal group synchronization problem~\eqref{eq:model} under additive Gaussian noise and its convex relaxation~\eqref{def:sdp}. 
The SDP relaxation is tight, i.e., {the globally optimal solution $\BX$ to~\eqref{def:sdp} is rank-$d$ and can be factorized into $\BX=\BS\BS^{\top}$ where $\BS\in\RR^{nd\times d}$ equals the maximum likelihood estimator~\eqref{def:od}}, if
\begin{equation}\label{eq:sigma1}
\sigma < \frac{c_0 \sqrt{n}}{\sqrt{d} (\sqrt{d} + \sqrt{\log n})}
\end{equation}
with high probability for a small constant $c_0.$
Moreover, its solution is unique and satisfies
\[
\min_{\BQ\in\Od(d)}\max_{1\leq i\leq n}\|\widehat{\BG}_{i} - \BG_i\BQ\|_F \lesssim \sigma\sqrt{n^{-1}d}(\sqrt{d} + \sqrt{\log n})
\]
where $\widehat{\BG}_{i}$ is the MLE and the solution to~\eqref{def:sdp}.
\end{theorem}
Here we have several remarks on Theorem~\ref{thm:sdp}. The maximum likelihood estimator $\widehat{\BG}$ does not have an explicit form and is not equal to the ground truth $\BG$. The best known bound on $\sigma$ for the tightness before this result is $\sigma < c_1n^{1/4}/d^{3/4}$ for some small constant $c_1$ in~\cite{L20a}. The bound~\eqref{eq:sigma1} greatly improves the scaling on $n$ from $n^{1/4}$ to a near-optimal dependence $\sqrt{n}$. 
Note that in random matrix theory, it has been extensively studied when the top eigenvectors of $\BA$ are correlated with the planted signals (low-rank matrix), see e.g.~\cite{BN11,PWBM18}. In fact, once the noise level $\sigma$ reaches $\sigma > \sqrt{n/d}$, spectral methods fail to provide useful information about the planted signal. Moreover, spectral methods (PCA) achieve the optimal detection threshold under certain natural priors for the spike. 
Therefore, the optimal bound on $\sigma$ should scale like $\sqrt{n/d}$ up to a logarithmic factor and our bound differs from this threshold by a factor of $\sqrt{d}$.

Another relevant approach to solve the $\Od(d)$ synchronization uses spectral relaxation. Note that~\cite{L20b} implies that spectral relaxation provides a spectral estimator $\widehat{\BG}_{i,\spec}$ of $\BG_i$ which satisfies $\min_{\BQ\in \Od(d)} \max\| \widehat{\BG}_{i,\spec} - \BG_i\BQ\| \lesssim \sigma \sqrt{ n^{-1}d}$ with a near-optimal bound on $\sigma$, i.e., $\sigma \lesssim\sqrt{n}/(\sqrt{d} + \sqrt{\log n}).$
However, the solution $\widehat{\BG}_{\spec}$ given by spectral relaxation is often not equal to the maximum likelihood estimator (MLE). %In fact, the spectral estimator equals the intialization $\BS^0$ and generalized power method improves the estimation in the later updates.

The second result concerns establishing an efficient algorithm with guaranteed global convergence. Despite that the spectral methods fail to give the MLE directly, it can be used as an initialization step since $\BS^0 = \widehat{\BG}_{\spec}$ is very close to $\BG$ even in block-wise error bound. Using this idea, we have the following theorem which provides a convergence analysis of the generalized power method.

\begin{theorem}\label{thm:main}
Consider orthogonal group synchronization problem~\eqref{eq:model} with additive Gaussian noise. 
Suppose 
\[
\sigma < \frac{c_0 \sqrt{n}}{\sqrt{d} (\sqrt{d} + \sqrt{\log n})} 
\]
for some small constant $c_0$,
{then the spectral initialization gives $\BS^0$ which satisfies $d_F(\BS^0,\BG) \leq \eps\sqrt{nd}$ for $\eps<1/(32\sqrt{d})$ and 
the sequence $\{\BS^t\}_{t=0}^{\infty}$ from the generalized power method converges linearly to a limiting point $\BS^{\infty}$, i.e., 
\[
d_F(\BS^t, \BS^{\infty}) \leq 2^{-t} d_F(\BS^0,\BS^{\infty}), \quad \forall t\geq 0,
\]
with probability at least $1 - O(n^{-2})$. Moreover, $\BS^{\infty}(\BS^{\infty})^{\top}$ is the unique global maximizer of the SDP relaxation~\eqref{def:sdp} and $\BS^{\infty}$ equals the maximum likelihood estimator, i.e., the global maximizer to~\eqref{def:od}.}
\end{theorem}
Theorem~\ref{thm:main} is a generalization of the convergence analysis of the GPM on phase synchronization (U(1) group)~\cite{ZB18}.  { This extension is nontrivial since U(1) group considered in~\cite{ZB18} is commutative while $\Od(d)$ group is non-commutative for $d\geq 3$. The resulting technical difference can be seen later in Lemma~\ref{lem:L} and~\ref{lem:phase} which are used to show that the GPM behaves like a contraction mapping on the iterates.}

Now we discuss a few future directions. 
As discussed before, the dependence on $n$ is near-optimal but the scaling of $d$ remains suboptimal by $\sqrt{d}$.
In fact, the suboptimal dependence of performance bound on the rank $d$ is a general issue in the convergence analysis of nonconvex approaches in many signal processing and machine learning problems~\cite{CCFMY20,LMCC19}. It is an open problem to have an exact recovery guarantee of the MLE via convex/nonconvex approach provided that
\[
\sigma = O\left(\sqrt{\frac{n}{d}}\right)
\]
holds modulo a log-factor.  Resolving this open problem will potentially lead to a substantial improvement of performance guarantees in many examples of low-rank matrix recovery via nonconvex approach. 
For the $\Od(d)$ synchronization, one possible solution is to derive a performance bound in terms of $d_{2}(\BX,\BY): = \min_{\BQ\in\Od(d)}\|\BY - \BX\BQ\|$, which is significantly more challenging since the norm is no longer equipped with an inner product. This will create technical difficulties in proving that the Algorithm 1 is actually a contraction mapping within a basin of attraction, more precisely discussed in Lemma~\ref{lem:lx} and~\ref{lem:con}. 

{ Another future direction is about the optimization landscape associated with the Burer-Monteiro factorization of~\eqref{def:od}. Empirical studies in the author's earlier work~\cite{L20a} indicate that even without spectral initialization, the generalized power method still outputs the global maximizer to the SDP relaxation. Therefore, it would be very interesting to see how the optimization landscape changes with respect to the noise strength $\sigma$ or to justify why random initialization works well enough for this synchronization problem.}
For the algorithmic aspect, another commonly-used approach for the angular (phase) synchronization is the GPM with a certain stepsize~\cite{B16,LYS17}. Instead of running $\BS^{t+1} = \PP_n(\BA\BS^t)$, one wants to recover the group elements by
\[
\BS^{t+1} = \PP_n\left( (\I_{nd} + \alpha \BA/n)\BS^t\right)
\]
for some $\alpha>0$. In particular, as $\alpha\rightarrow\infty$, this algorithm becomes the generalized power method studied in this paper. It would be an interesting research problem to explore whether one can establish the global convergence of this proposed approach. 

Finally, note that our current manuscript only deals with the additive Gaussian noise in~\eqref{eq:model}. It is natural to consider general subgaussian noise. Examples include uniform corruption model~\cite{WS13}, i.e., $\BA_{ij} = \xi_{ij} \BG_i\BG_j^{\top} + (1-\xi_{ij})\BO_{ij}$, $i<j$ where $\xi_{ij}$ is an independent Bernoulli($p$) random variable and $\BO_{ij}$ is a random matrix sampled uniformly from the Haar measure on $\Od(d)$, 
and also general additive block-wise independent subgaussian noise, i.e., $\BA_{ij} = \BG_i\BG_j^{\top} + \sigma\BDelta_{ij}$ where $\{\BDelta_{ij}\}_{i\leq j}$ are centered independent random matrices with subgaussian tail on $\|\BDelta_{ij}\|$. The techniques introduced in this paper could be extended to both scenarios but a tight performance bound requires a more careful analysis of the operator norm of the noise matrix.

\section{Proofs}\label{s:proof}

This section is devoted to the proof of the main theorems. Our proof will follow a similar route in~\cite{ZB18}. However, since $\Od(d)$ is non-commutative, many technical parts need different treatments from those in~\cite{ZB18}. {Also we will use some supporting results in~\cite{L20b}}, especially regarding spectral initialization, without giving the detailed proofs. To make this manuscript more self-contained, the proof idea of each cited result will be discussed.  

{Instead of studying~\eqref{eq:model}, we will focus on a statistically equivalent model.}
Note that orthogonal transform will not change the distribution of Gaussian random matrix. It suffices to consider
\[
\BA = \BZ\BZ^{\top} + \sigma \BW \in\RR^{nd\times nd}
\]
after performing a simple change of variable to~\eqref{eq:model} where $\BZ^{\top} := [\I_d,\cdots,\I_d]\in\RR^{d\times nd}.$ 

Our main goal is to show the tightness of the SDP relaxation~\eqref{def:sdp}, i.e., the global maximizer to the SDP relaxation is of rank $d$, which corresponds to the global maximizer to~\eqref{def:od}. Similar to phase synchronization, the global maximizer to~\eqref{def:od} does not have an explicit form ($\BZ$ is not the global maximizer in general), which makes it harder to characterize the global maximum. To resolve this issue,~\cite{ZB18} developed a smart way of characterizing the global maximizer by treating it as the limiting point of the sequences generated by the generalized power method. We will show that the sequence $\{\BS^{\ell}\}_{\ell=0}^{\infty}$ always stays in the intersection of the following two sets in Section~\ref{ss:conmap} and~\ref{ss:basin}:
{
\begin{align*}
{\cal N}_{\eps} & := \left\{\BS\in \Od(d)^{\otimes n}: d_F(\BS,\BZ)  \leq \eps \sqrt{nd} \right\}, \\
{\cal N}_{\xi,\infty} & : = \left\{\BS\in \Od(d)^{\otimes n}: \left\| \BW_i^{\top}\BS\right\|_F \leq \xi\sqrt{nd}(\sqrt{d} + 4\sqrt{\log n}),~1\leq i\leq n \right\},
\end{align*} }
where { $\xi = 2+3\eps > 0$ is a constant determined in~\eqref{def:xi}}, $\eps$ is another constant which satisfies
\[
\eps \leq \frac{1}{32\sqrt{d}} \Longleftrightarrow \eps^2 d \leq \frac{1}{1024}
\]
and $\BW_i\in\RR^{nd\times d}$ is the $i$th block column of noise matrix $\BW = [\BW_1,\cdots,\BW_n].$ 

Note that both $\mathcal{N}_{\eps}$ and $\mathcal{N}_{\xi,\infty}$ are closed and compact sets in the Euclidean space. Then if we can prove that $\{\BS^{\ell}\}_{\ell=0}^{\infty}$ converges to a limiting point, the limit also stays in $\mathcal{N}_{\eps}\cap\mathcal{N}_{\xi,\infty}$. Moreover, we have the following theorem.

\begin{theorem}\label{thm:optlimit}
Suppose the generalized power iteration converges to a limiting point $\BS^{\infty}$  which satisfies $\BA\BS^{\infty}= \BLambda\BS^{\infty}$ where $\BLambda = \blkdiag(\BLambda_{11},\cdots,\BLambda_{nn})\succeq 0$ is a symmetric block-diagonal matrix with $\BLambda_{ii}\in\RR^{d\times d}$. Moreover, if $\BS^{\infty}$ is located in $\mathcal{N}_{\eps}\cap \mathcal{N}_{\infty}$, then $\BX^{\infty} =\BS^{\infty}(\BS^{\infty})^{\top} $ is the unique global maximizer to~\eqref{def:sdp} if
\begin{equation}\label{eq:sigma}
\sigma < \frac{ 1 - 3\eps^2 d/2}{ \xi + 3}  \cdot \frac{\sqrt{n}}{\sqrt{d}(\sqrt{d} + 4\sqrt{\log n}) }.
\end{equation}
\end{theorem}

Therefore, our task becomes (a): to show that the sequence $\{\BS^{\ell}\}_{\ell=0}^{\infty}$ is located in the $\mathcal{N}_{\eps}\cap\mathcal{N}_{\xi,\infty}$ and converges to a limiting point; (b): to justify why the limiting point $\BS^{\infty}$ is the global maximizer. 

Now we first focus on completing the task (b) by providing a sufficient condition of the global optimality in~\eqref{def:sdp}. Then we will show that if the limiting point is a fixed point of the power iteration and is located in $\mathcal{N}_{\eps}\cap\mathcal{N}_{\xi,\infty}$, then it must satisfy the following global optimality condition. Finally, we will verify the assumptions  in Theorem~\ref{thm:optlimit} hold in Section~\ref{ss:main_proof}.

\begin{theorem}\label{thm:cvx}
The matrix $\BX = \BS\BS^{\top}$ with $\BS\in \Od(d)^{\otimes n}$ is a global optimal solution to~\eqref{def:sdp} if there exists a block-diagonal matrix $\BLambda\in\RR^{nd\times nd}$ such that
\begin{equation}\label{cond:opt}
\BA\BS = \BLambda \BS, \quad \BLambda - \BA\succeq 0.
\end{equation}
Moreover, if $\BLambda - \BA$ is of rank $(n-1)d$, then $\BX$ is the unique global maximizer. 
\end{theorem}

This optimality condition can be found in several places including~\cite[Proposition 5.1]{L20a} and~\cite[Theorem 7]{RCBL19}.
The derivation of Theorem~\ref{thm:cvx} follows from the standard routine of duality theory in convex optimization. {In fact, the block-diagonal matrix $\BLambda$ corresponds exactly to the dual variable in~\eqref{def:sdp}.} The first equation in~\eqref{cond:opt} seems relatively simple to achieve as we have seen that the fixed point of power iteration naturally satisfies this condition~\eqref{eq:fixed}.  The more challenging part arises from achieving $\BLambda - \BA\succeq 0$ which reduces to finding a tight lower bound for the smallest eigenvalue of $\BLambda.$ We will prove that a tight lower bound of $\lambda_{\min}(\BLambda)$ can be obtained if the limiting point is inside $\mathcal{N}_{\eps}\cap\mathcal{N}_{\xi,\infty}.$ Before proceeding to the proof of Theorem~\ref{thm:optlimit}, we introduce two useful supporting facts.

\begin{lemma}
Suppose $\BW\in\RR^{nd\times nd}$ is a symmetric Gaussian random matrix, then
\begin{equation}\label{eq:gaussnorm}
\|\BW\| \leq 3\sqrt{nd}
\end{equation}
with probability at least $1 - \exp(-nd/2).$
\end{lemma}
This is a quite standard result in random matrix theory, which can be found in many places~\cite{BBS17,ZB18,V18}.

\begin{lemma}\label{lem:sigmamin}
For any $\BS\in\mathcal{N}_{\eps}$, all the singular values of $\BZ^{\top}\BS$ satisfy
\[
\left( 1 - \frac{\eps^2d}{2}\right)n \leq \sigma_{i}(\BZ^{\top}\BS)\leq n, \quad 1\leq i\leq d.
\]
\end{lemma}

\begin{proof}
Note that
\begin{align*}
d_F^2 (\BS,\BZ) & = \|\BS - \BZ\BQ\|_F^2 \\
& = \|\BS\|_F^2 + \|\BZ\|_F^2 - 2\lag \BZ^{\top}\BS, \BQ\rag \\
&  = 2(nd - \|\BZ^{\top}\BS\|_*) \leq \eps^2 nd
\end{align*}
where $\BQ = \PP(\BZ^{\top}\BS)$. Since $\|\BZ^{\top}\BS\| \leq n$ implies $0\leq \sigma_i(\BZ^{\top}\BS) \leq n$, $1\leq i\leq d$ , we have
\[
\frac{\eps^2nd}{2} \geq \sum_{i=1}^d (n - \sigma_i(\BZ^{\top}\BS)) \geq n - \sigma_{\min}(\BZ^{\top}\BS)\Longrightarrow \sigma_{\min}\left(\BZ^{\top}\BS\right) \geq \left(1 - \frac{\eps^2 d}{2}\right)n.
\]
\end{proof}

This bound is actually quite conservative due to the extra $d$ factor in the lower bound. It is also the reason why we obtain a sub-optimal scaling on $\sigma$ in the main theorem. { In fact, numerical experiments indicate that this additional $d$ factor is not tight. Therefore, we leave this sub-optimality issue to the future work.}
Now we are ready to prove Theorem~\ref{thm:optlimit}.

\begin{proof}[\bf Proof of Theorem~\ref{thm:optlimit}]
Suppose the fixed point $\BS^{\infty}\in\Od(d)^{\otimes n}$ satisfies $\BA\BS^{\infty} = \BLambda\BS^{\infty}$ for some positive semidefinite block-diagonal matrix $\BLambda = \blkdiag(\BLambda_{11},\cdots,\BLambda_{nn})\in\RR^{nd\times nd}$. Then we have
\[
\BLambda_{ii} = [\BA\BS^{\infty}]_i (\BS_i^{\infty})^{\top}  \succeq 0, \quad 1\leq i\leq n,
\]
{ where $[\BA\BS^{\infty}]_i = \BLambda_{ii}\BS_i^{\infty},1\leq i\leq n$. Remember our goal is to prove that $\BLambda -\BA\succeq 0$ and its $(d+1)$th smallest eigenvalue is strictly positive.}

We first obtain a lower bound of the smallest eigenvalue of $\BLambda_{ii}$ as follows:
\begin{align*}
\lambda_{\min}(\BLambda_{ii}) & = \sigma_{\min}(\BLambda_{ii})  = \sigma_{\min}\left( \sum_{j=1}^n \BA_{ij}\BS_j^{\infty}\right)\\ 
& = \sigma_{\min}\left( \BZ^{\top}\BS^{\infty} + \sigma\sum_{j=1}^n \BW_{ij}\BS_j^{\infty}\right) \\
\text{(Weyl's theorem)} \quad & \geq \sigma_{\min}\left( \BZ^{\top}\BS^{\infty}\right) - \sigma\left\| \BW_i^{\top}\BS^{\infty}\right\|_F \\
\text{(Lemma~\ref{lem:sigmamin})}\quad & \geq \left( 1 - \frac{\eps^2 d}{2}\right)n - \sigma\xi\sqrt{nd}(\sqrt{d} + 4\sqrt{\log n})
\end{align*}
where $\BLambda_{ii} \succeq 0$ is symmetric for all $1\leq i\leq n$ and $\BS^{\infty}\in \mathcal{N}_{\eps}$. Note that $(\BLambda -\BA)\BS^{\infty} = 0$ means at least  $d$ eigenvalues are 0. To show that $\BLambda - \BA\succeq 0$ with rank $(n-1)d$, {it suffices to ensure the $(d+1)$-th smallest eigenvalue of $\BLambda -\BA$ is strictly positive, i.e., $\bu^{\top}(\BLambda - \BA)\bu > 0$ for any nonzero unit vector $\bu \in\RR^{nd}$ such that $\bu^{\top}\BS^{\infty} = 0$.}  It holds that
\begin{align*}
\bu^{\top}\BA\bu & = \bu^{\top}\BZ\BZ^{\top}\bu + \sigma\bu^{\top}\BW\bu \\
& =  \|\bu^{\top}(\BS^{\infty}  - \BZ\BQ)\|^2 + \sigma\bu^{\top}\BW\bu \\
& \leq \|\BS^{\infty}  - \BZ\BQ\|^2 + \sigma \|\BW\| \\
(\text{Use}~\eqref{eq:gaussnorm})\quad& \leq \eps^2 nd + 3\sigma \sqrt{nd}
\end{align*}
where $\|\BW\| \leq 3\sqrt{nd}$, $\|\BS^{\infty} - \BZ\BQ\|\leq \|\BS^{\infty} - \BZ\BQ\|_F\leq \eps\sqrt{nd}$ and $\BQ = \PP(\BZ^{\top}\BS^{\infty})$. Thus if~\eqref{eq:sigma} holds, we have
\begin{align*}
\bu^{\top}(\BLambda - \BA)\bu & \geq \min_{1\leq i\leq n}\lambda_{\min}(\BLambda_{ii}) - \max_{\bu:\bu^{\top}\BS^{\infty} = 0,\|\bu\| = 1} \bu^{\top}\BA\bu \\
& \geq  \left( 1 - \frac{\eps^2 d}{2}\right)n - \sigma\xi\sqrt{nd}(\sqrt{d} + 4\sqrt{\log n}) - \left(\eps^2 nd + 3\sigma \sqrt{nd}\right) \\
& \geq \left( 1 - \frac{3\eps^2 d}{2}\right)n - \sigma (\xi + 3)\sqrt{nd} \left(\sqrt{d} + 4\sqrt{\log n} \right) > 0.
\end{align*}
Then we have $\BLambda - \BA \succeq 0$ with the $(d+1)$th smallest eigenvalue strictly positive and the bottom $d$ eigenvalues equal to zero.
\end{proof}

{\bf \underline{\underline{Roadmap}}:} Now we provide an overview of what to do in the next few sections. 
To prove the main results Theorem~\ref{thm:sdp} and~\ref{thm:main}, Theorem~\ref{thm:optlimit} indicates that we need to show the iterates $\{\BS^{\ell}\}_{\ell=0}^{\infty}$ stay in $\mathcal{N}_{\eps}\cap\mathcal{N}_{\xi,\infty}$ for all $\ell\in\mathbb{Z}_{\geq 0}$ and moreover, it will converge to a limiting point which satisfies $\BA\BS^{\infty} = \BLambda\BS^{\infty}$ for some positive semidefinite block-diagonal matrix $\BLambda\in\RR^{nd\times nd}.$ We will justify why $\{\BS^{\ell}\}_{\ell=0}^{\infty}$ stay in $\mathcal{N}_{\eps}\cap\mathcal{N}_{\xi,\infty}$ in Section~\ref{ss:conmap} and~\ref{ss:basin}. The existence and convergence of limiting point, and its global optimality will be given in Section~\ref{ss:main_proof}. In Section~\ref{ss:init}, we will prove why the spectral initialization provides a good starting point within $\mathcal{N}_{\eps}\cap\mathcal{N}_{\xi,\infty}$. All the following results hold with high probability under the assumptions in Theorem~\ref{thm:sdp} and~\ref{thm:main}.

\subsection{Contraction mapping on $\mathcal{N}_{\eps}\cap\mathcal{N}_{\xi,\infty}$}\label{ss:conmap}
The tricky part is that we do not know the explicit form of the MLE which is not necessarily equal to $\BZ$. In order to show the convergence of iterates to the MLE, we start with proving the power iteration essentially gives a contraction mapping on $\mathcal{N}_{\eps}\cap \mathcal{N}_{\xi,\infty}$. 
Define
\begin{equation}\label{def:L}
{\cal L} := n^{-1}\BA, \quad {\cal L}\BS := \frac{1}{n}\BA\BS = \frac{1}{n} (\BZ\BZ^{\top} + \sigma\BW)\BS \in\RR^{nd\times d}.
\end{equation}
Then the power iteration is equivalent to:
\[
\BS^{t+1} = \PP_n ({\cal L} \BS^t).
\]

Now, we will prove that $\PP_n\circ {\cal L}$ is a contraction mapping, by investigating ${\cal L}$ and $\PP_n$ individually. 
{ The main result in this subsection is Lemma~\ref{lem:con}.} 
Ideally, if we manage to prove that $\PP_n\circ {\cal L}$ maps $\mathcal{N}_{\eps}\cap \mathcal{N}_{\xi,\infty}$ to itself, then the proof is done which follows from the contraction mapping theorem. However, it remains unclear how to prove/disprove this statement. Instead, we will show that the iterates $\{\BS^{\ell}\}_{\ell=0}^{\infty}$ from the projected power method stay in the basin of attraction. We will address this issue in the next subsection.

\begin{lemma}\label{lem:L}
Suppose $\BX$ and $\BY\in \Od(d)^{\otimes n}$ with $d_F(\BX,\BZ) \leq \eps\sqrt{nd}$ and $d_F(\BY,\BZ)\leq \eps\sqrt{nd}$, it holds 
\[
d_F({\cal L}\BX, {\cal L}\BY)  \leq\rho'\cdot d_F(\BX,\BY),
\]
where
\[
\rho' = 2\eps\sqrt{d} + 3\sigma\sqrt{\frac{d}{n}}.
\]
Under the assumption
\[
\eps\sqrt{d} \leq \frac{1}{32}, \quad\sigma \leq\frac{1}{72}\sqrt{\frac{n}{d}},
\]
then
\[
\rho' \leq \frac{1}{9}.
\]
\end{lemma}
We give the proof of this lemma at the end of this subsection. {Note that Lemma~\ref{lem:L} generalizes Lemma 12 in~\cite{ZB18} to the non-commutative scenario. The major difficulty in proving Lemma~\ref{lem:L} comes from estimating the error $\|\BZ^{\top}(\BX-\BY\BQ)\|$ for $\BQ =\PP(\BY^{\top}\BX)$, which is very different from the proof of Lemma 12 in~\cite{ZB18}.} 
Next, we turn our focus to $\PP_n(\cdot)$.  
\begin{lemma}\label{lem:phase}
For two invertible matrices $\BX\in\RR^{d\times d}$ and $\BY\in\RR^{d\times d}$, we have
\begin{align*}
\|\PP(\BX)  - \PP(\BY)\|_F & \leq \frac{2\|\BX - \BY\|_F}{\sigma_{\min}(\BX) + \sigma_{\min}(\BY)}.
\end{align*}
If only $\sigma_{\min}(\BX)$ (or $\sigma_{\min}(\BY)$) is known, then we can use
\[
\|\PP(\BX)  - \PP(\BY)\|_F \leq 2\sigma^{-1}_{\min}(\BX)\|\BX - \BY\|_F.
\]
\end{lemma}
{Lemma~\ref{lem:phase} is a generalization of Lemma 13 in~\cite{ZB18} from $d\leq 2$ to $d\geq 3$, which says the matrix sign function (generalized phase) in~\eqref{def:P} is a Lipschitz continuous function in the region where the input matrices have their smallest singular values bounded away from 0.
We defer the proof of Lemma~\ref{lem:phase} to the end of this section. The main idea of the proof follows from applying the famous Davis-Kahan theorem to the augmented matrices of $\BX$ and $\BY$.} 

Now it seems combining Lemma~\ref{lem:L} with~\ref{lem:phase} justifies the contraction of $\PP_n\circ {\cal L}$. The only missing piece is to ensure that each block $[{\cal L}\BX]_i$ has a lower bound on its smallest singular value so that we can apply Lemma~\ref{lem:phase}.

\begin{lemma}\label{lem:lx}
Suppose $ \BX\in{\cal N}_{\eps} \cap {\cal N}_{\xi,\infty}$, then 
\[
\sigma_{\min}\left( [{\cal L}\BX]_i\right) \geq 1 - \eps^2 d
\]
provided that
\[
\sigma \leq \frac{\eps^2 d}{2\xi }\cdot \frac{ \sqrt{n}}{\sqrt{d} (\sqrt{d} + \sqrt{\log n})}.
\]
\end{lemma}
One can easily realize that the dependence of $\sigma$ on $d$ is sub-optimal by a factor of $\sqrt{d}.$ This is where the bottleneck is.

\begin{lemma}[Contraction mapping]\label{lem:con}
Suppose $\BX$ and $\BY$ are in $\mathcal{N}_{\eps}\cap\mathcal{N}_{\xi,\infty}$, then we have
\begin{align*}
d_F(\PP_n({\cal L}\BX), \PP_n({\cal L}\BY)) & \leq \frac{2}{1-\eps^2d} \cdot d_F( {\cal L}\BX, {\cal L}\BY) \\
& \leq \frac{2}{1-\eps^2d} \left(2\eps\sqrt{d} + 3\sigma\sqrt{\frac{d}{n}} \right)\cdot d_F(\BX,\BY).
\end{align*}
We define
\[
\rho =  \frac{2}{1-\eps^2d} \left(2\eps\sqrt{d} + 3\sigma\sqrt{\frac{d}{n}} \right).
\]
Under
\[
\eps\sqrt{d} \leq \frac{1}{32}, \quad\sigma \leq \min\left\{ \frac{1}{72}\sqrt{\frac{n}{d}},~ \frac{\eps^2d }{2\xi }\cdot  \frac{\sqrt{n}}{\sqrt{d}(\sqrt{d} + 4\sqrt{\log n})} \right\},
\]
then $\rho < 1/2.$

\end{lemma}

\begin{proof}[\bf Proof of Lemma~\ref{lem:lx} and~\ref{lem:con}]
Over $\BX\in \mathcal{N} = {\cal N}_{\eps} \cap {\cal N}_{\xi,\infty}$, we have ${\cal L}\BX$ and its $i$th block as
\[
{\cal L}\BX = \frac{1}{n} (\BZ \BZ^{\top} + \sigma\BW)\BX, \quad
[{\cal L}\BX ]_i = \frac{1}{n}\BZ^{\top}\BX + \frac{\sigma}{n}\BW_i^{\top}\BX.
\]
Note that if $d_F(\BX,\BZ)\leq \eps\sqrt{nd}$, $
\sigma_{\min}(\BZ^{\top}\BX) \geq \left(1 - \eps^2d/2\right)n$ holds.
As a result,
\begin{align*}
\sigma_{\min}([{\cal L}\BX]_i) & \geq \frac{1}{n}\sigma_{\min}(\BZ^{\top}\BX) - \frac{\sigma}{n} \left\| \BW_i^{\top}\BX\right\| \\
(\text{Lemma}~\ref{lem:sigmamin},~\BX\in\mathcal{N}_{\xi,\infty})\quad & {\geq } \left( 1 - \frac{\eps^2d}{2}\right)  - \frac{\sigma}{n}\cdot \xi\sqrt{nd}(\sqrt{d} + 4\sqrt{\log n})  \\
&  \geq 1 - \eps^2d
\end{align*}
when 
\[
\sigma \leq \frac{\eps^2 \sqrt{nd}}{2\xi (\sqrt{d} + 4\sqrt{\log n})}.
\]
Similarly, we have $\sigma_{\min}([{\cal L}\BY]_i) \geq 1-\eps^2 d.$

Now we are ready to obtain the contraction property of $\PP_n\circ {\cal L}$ on $\mathcal{N}_{\eps}\cap\mathcal{N}_{\xi,\infty}$. We pick $\BQ$ as the orthogonal matrix $\BQ$ which minimizes 
\[
d_F({\cal L}\BX, {\cal L}\BY) = \min_{\BQ\in \Od(d)} \|{\cal L}\BX - {\cal L}\BY\BQ \|_F.
\]
Then
\begin{align*}
d_F(\PP_n({\cal L}\BX), \PP_n({\cal L}\BY)) &\leq  \|\PP_n({\cal L}\BX) - \PP_n({\cal L}\BY\BQ)\|_F\\
 & =  \sqrt{ \sum_{i=1}^n \| \PP([{\cal L}\BX]_i) - \PP([{\cal L}\BY\BQ]_i) \|_F^2} \\
\text{(Lemma~\ref{lem:phase})} \quad& \leq \frac{2}{1-\eps^2 d} \sqrt{ \sum_{i=1}^n \| [{\cal L}\BX]_i - [{\cal L}\BY\BQ]_i) \|_F^2 } \\
& = \frac{2}{1-\eps^2 d}  \| {\cal L}\BX -{\cal L} \BY\BQ \|_F \\
& = \frac{2}{1-\eps^2 d} d_F({\cal L}\BX,{\cal L}\BY)\\
\text{(Lemma~\ref{lem:L})} \quad &  \leq  \frac{2}{1-\eps^2 d}\left(2\eps\sqrt{d} + 3\sigma\sqrt{\frac{d}{n}} \right) d_F(\BX,\BY).
\end{align*}

\end{proof}

{Now we present the proof of Lemma~\ref{lem:L} and~\ref{lem:phase} which are used to justify Lemma~\ref{lem:con}.}

\begin{proof}[\bf Proof of Lemma~\ref{lem:L}]
Under $d_F(\BX,\BZ) \leq \eps\sqrt{nd}$ and $d_F(\BY,\BZ)\leq \eps\sqrt{nd}$, it holds that
\[
\sigma_{\min}(\BX^{\top}\BZ) \geq \left(1 - \frac{\eps^2d}{2} \right)n, \quad \sigma_{\min}(\BY^{\top}\BZ) \geq \left(1 - \frac{\eps^2d}{2} \right)n.
\]
Also by triangle inequality and Lemma~\ref{lem:sigmamin}, we have
\[
d_F(\BX,\BY)\leq 2\eps\sqrt{nd}, \quad \sigma_{\min}(\BX^{\top}\BY) \geq \left(1 - 2\eps^2 d\right)n.
\]

Now we set $\BQ = \PP(\BY^{\top}\BX)$ which minimizes $\|\BX - \BY\BR\|_F$ over $\BR\in\Od(d)$.
\begin{align*}
d_F({\cal L}\BX, {\cal L}\BY) & \leq \frac{1}{n}\| \BA(\BX - \BY\BQ)  \|_F \\
& \leq \frac{1}{n} \left( \left\|\BZ\BZ^{\top}(\BX - \BY\BQ)\right\|_F + \sigma\|\BW (\BX - \BY\BQ)\|_F \right) \\
& \leq \frac{1}{\sqrt{n}} \|\BZ^{\top} (\BX - \BY\BQ)\|_F + 3\sigma\sqrt{\frac{d}{n}} \cdot d_F(\BX,\BY)
\end{align*}
where $\BA = \BZ\BZ^{\top} + \sigma\BW$, $\|\BZ\| = \sqrt{n}$, and $\|\BW\|\leq 3\sqrt{nd}.$
%Note that for $\sigma \leq \eta\sqrt{\frac{n}{d}}$ for some $\eta > 0$, the second term above is small. 

The more tricky part is the estimation of $\| \BZ^{\top}(\BX - \BY\BQ) \|_F.$
Let $\BQ = \PP(\BY^{\top}\BX)$, $\BQ_{\BX} = \PP(\BZ^{\top}\BX)$, and $\BQ_{\BY} = \PP(\BZ^{\top}\BY)$. Then
\begin{align*}
\|\BZ^{\top} (\BX - \BY\BQ) \|_F & \leq \|( \BY - \BZ\BQ_{\BY})^{\top} (\BX - \BY\BQ ) \|_F  + \| \BY^{\top} (\BX - \BY\BQ) \|_F  \\
& \leq \eps\sqrt{nd} \|\BX - \BY\BQ \|_F + \| \BY^{\top}\BX - n\BQ \|_F
\end{align*}
where $\|\BY - \BZ\BQ_{\BY}\| \leq \|\BY - \BZ\BQ_{\BY}\|_F\leq \eps\sqrt{nd}.$
For the second term,  {
\begin{align*}
 \|\BY^{\top}\BX - n\BQ\|_F^2 & = \|\BY^{\top}\BX\|_F^2 - 2n\|\BY^{\top}\BX\|_* + n^2d \\
& =\sum_{i=1}^d (n - \sigma_i(\BY^{\top}\BX))^2 \\
& \leq (n - \sigma_{\min}(\BY^{\top}\BX))\cdot (nd - \|\BY^{\top}\BX\|_*) \\
& \leq \eps^2 nd \|\BX - \BY\BQ\|_F^2
\end{align*}}
where $0\leq n - \sigma_{\min}(\BY^{\top}\BX) \leq 2\eps^2 nd$ follows from Lemma~\ref{lem:sigmamin} and $\|\BX - \BY\BQ\|_F^2 = 2(nd - \|\BY^{\top}\BX\|_*)$.
As a result,
\[
\|\BZ^{\top} (\BX - \BY\BQ) \|_F  \leq 2\eps \sqrt{nd} \|\BX - \BY\BQ \|_F = 2\eps\sqrt{nd}\cdot d_F(\BX,\BY).
\]
To sum up, we have 
\begin{align*}
d_F({\cal L}\BX, {\cal L}\BY) & \leq \frac{1}{\sqrt{n}}\cdot 2\eps\sqrt{nd}\cdot d_F(\BX,\BY) + 3\sigma\sqrt{\frac{d}{n}}\cdot d_F(\BX,\BY)\\
& \leq \left(2\eps\sqrt{d} + 3\sigma\sqrt{\frac{d}{n}} \right)\cdot d_F(\BX,\BY).
\end{align*}
\end{proof}

Now we present the proof of Lemma~\ref{lem:phase} which relies on the famous Davis-Kahan theorem. 
\begin{theorem}[Davis-Kahan theorem~\cite{DK70}]\label{thm:DK}
Suppose $\BPsi$ and $\BPsi_E$ are the top $d$ normalized eigenvectors of $\BX$ and $\BX_{\BE} = \BX + \BE$ respectively, then
\[
\min_{\BQ\in\Od(d)}\| \BPsi_{\BE} - \BPsi\BQ \| \leq \sqrt{2} \| (\I - \BPsi\BPsi^{\top})\BPsi_E \| \leq \frac{\sqrt{2}\|\BE\BPsi_E\|}{\delta}
 %\| (\I - \BPsi\BPsi^{\top})\BPsi_E \| \leq \frac{\|\BE\BPsi_E\|}{\delta}
\]
where $\BPsi^{\top}\BPsi = \BPsi_E^{\top}\BPsi_E = \I_d$ and
$\delta =  \lambda_{d}(\BX) - \lambda_{d+1}(\BX_{\BE})>0$ is the spectral gap. The bound also holds when replacing $\|\cdot\|$ by $\|\cdot\|_F$.
\end{theorem}
In particular, in order to obtain a meaningful bound, we would hope to have $\delta > \|\BE\|$ such that $\|(\I - \BPsi\BPsi^{\top})\BPsi_{\BE}\| < 1.$ 

By Theorem~\ref{thm:DK}, we can establish a useful result about the projection operator $\PP(\cdot).$

\begin{proof}[\bf Proof of Lemma~\ref{lem:phase}]
The proof follows from the Davis-Kahan theorem. 
Consider the augmented matrix:
\[
\widetilde{\BX} = 
\begin{bmatrix}
0 & \BX \\
\BX^{\top} & 0
\end{bmatrix}, \quad 
\widetilde{\BY} = 
\begin{bmatrix}
0 & \BY \\
\BY^{\top} & 0
\end{bmatrix}.
\]
Let $(\BU_X,\BV_X)$ and $(\BU_Y,\BV_Y)$ be the left and right singular vectors of $\BX$ and $\BY$ respectively. Then 
$\BM_X := \frac{1}{\sqrt{2}}\begin{bmatrix}
\BU_X \\
\BV_X
\end{bmatrix}$ and $\BM_Y := \frac{1}{\sqrt{2}}\begin{bmatrix}
\BU_Y \\
\BV_Y
\end{bmatrix}$ are the eigenvectors associated with the top $d$ eigenvalues of $\widetilde{\BX}$ and $\widetilde{\BY}$ respectively which are also equal to the singular values of $\BX$ and $\BY$. 
{Applying the Davis-Kahan Theorem~\ref{thm:DK} gives
\begin{align*}
\| (\I_{2d} - \BM_{\BX}\BM_{\BX}^{\top})\BM_{\BY}\|_F
& \leq \frac{1}{\lambda_{d}(\widetilde{\BX})-\lambda_{d+1}(\widetilde{\BY})} \cdot\| (\widetilde{\BX} - \widetilde{\BY})\BM_{\BY} \|_F \\
& \leq  \frac{1}{\sigma_{\min}(\BX)+ \sigma_{\min}(\BY)} \cdot \|\BX - \BY\|_F\|\BM_{\BY}\| \\
& \leq  \frac{ \|\BX - \BY\|_F}{\sigma_{\min}(\BX)+ \sigma_{\min}(\BY)}.
\end{align*}
Here the spectral gap equals $\sigma_{\min}(\BX)+\sigma_{\min}(\BY)$ since all the eigenvalues of $\widetilde{\BX}$ and $\widetilde{\BY}$ are $\{\pm \sigma_{\ell}(\BX)\}_{\ell=1}^{d}$ and $\{\pm \sigma_{\ell}(\BY)\}_{\ell=1}^{d}$ respectively. Therefore, we have $\lambda_{d}(\widetilde{\BX}) = \sigma_{\min}(\BX)$ and $\lambda_{d+1}(\widetilde{\BY}) = - \sigma_{\min}(\BY).$
Using the definition of $\BM_X$ and $\BM_Y$, we have
\begin{align*}
(\I_{2d} - \BM_{\BX}\BM_{\BX}^{\top})\BM_{\BY} & = \frac{1}{2\sqrt{2}}
\begin{bmatrix}
\I_d & -\BU_X\BV_X^{\top} \\
-\BV_X\BU_X^{\top} & \I_d
\end{bmatrix}  
\begin{bmatrix}
\BU_Y \\
\BV_Y
\end{bmatrix} \\
& = \frac{1}{2\sqrt{2}} 
\begin{bmatrix}
-\I_d & 0\\
0 & \BV_X\BU_X^{\top}
\end{bmatrix}
\begin{bmatrix}
\BU_X\BV_X^{\top} - \BU_Y\BV_Y^{\top} \\
\BU_X\BV_X^{\top} - \BU_Y\BV_Y^{\top}
\end{bmatrix}
\BV_Y.
\end{align*}
Note that $\BU_X,\BV_X,\BU_Y,$ and $\BV_Y$ are all orthogonal. Therefore, it holds that
\[
\|(\I_{2d} - \BM_{\BX}\BM_{\BX}^{\top})\BM_{\BY}\|_F = \frac{1}{2}\| \BU_X\BV_X^{\top} - \BU_Y\BV_Y^{\top}\|_F = \frac{1}{2}\| \PP(\BX) - \PP(\BY) \|_F.
\]

Now we obtain the distance between $\PP(\BX)$ and $\PP(\BY):$
\begin{align*}
\|\PP(\BX) - \PP(\BY)\|_F & = 2\|(\I_{2d} - \BM_{\BX}\BM_{\BX}^{\top})\BM_{\BY}\|_F  \leq \frac{2\|\BX -\BY\|_F}{\sigma_{\min}(\BX) + \sigma_{\min}(\BY)}.
\end{align*}}
\end{proof}

\subsection{Keep the iterates in the basin of attraction}\label{ss:basin}
The challenging part is to keep the iterates $\BS^{\ell}$ in the basin of attraction $\mathcal{N}_{\eps}\cap \mathcal{N}_{\xi,\infty}$. In particular, it is not trivial to estimate the ``correlation" between $\BW_m$ and $\BS^{\ell}$ since $\BS^{\ell}$ depends on $\BW_m$ implicitly for all $1\leq m\leq n$. The remedy is to use the very popular~\emph{leave-one-out} technique by introducing an auxiliary sequence  which is the output of the generalized power method based on the data matrix:
\[
\BA^{(m)}_{ij} :=  
\begin{cases}
\I_d + \sigma\BW_{ij}, & \text{ if } ~i\neq m \text{ and } j\neq m, \\
\I_d, & \text{ if }~ i = m \text{ or } j =m,
\end{cases}
\]
for every $1\leq m\leq n.$ The matrix $\BA^{(m)}$ is exactly equal to $\BA$ except its $m$th block row and column. In other words, 
\begin{equation}\label{eq:Am}
\BA^{(m)} := \BZ\BZ^{\top} + \sigma \BW^{(m)}
\end{equation}
where 
\[
\BW^{(m)}_{ij} =  
\begin{cases}
\BW_{ij}, & \text{ if } ~i\neq m \text{ and } j\neq m, \\
0, & \text{ if }~ i = m \text{ or } j =m.
\end{cases}
\]
Denote $\BS^{\ell,m}$ as the $\ell$th iterate from the generalized power method with 
\[
{\cal L}^{(m)} := n^{-1}\BA^{(m)}= n^{-1}\left( \BZ\BZ^{\top} + \sigma\BW^{(m)}\right),
\]
which means
\[
\BS^{\ell+1,m} = \PP_n (  {\cal L}^{(m)}  \BS^{\ell,m}), \quad 1\leq m\leq n, \quad \ell\geq 0.
\]

By construction, $\{\BS^{\ell,m}\}_{\ell=0}^{\infty}$ is~\emph{statistically independent} of $\BW_m$ and thus applying the concentration inequality to $\|\BW_m^{\top}\BS^{\ell,m}\|_F$ yields a tight bound. The missing piece is the difference between $\BS^{\ell,m}$ and $\BS^{\ell}$ for all $1\leq m\leq n$ and $\ell\geq 0.$ Note that $\BA$ and $\BA^{(m)}$ differ up to only one block row/column and thus we believe $\BS^{\ell}$ and $\BS^{\ell,m}$ are very close for all $\ell\in\mathbb{Z}_{\geq 0}$ and $1\leq m\leq n$.
As long as $d_F(\BS^{\ell}, \BS^{\ell,m})$ is small (in particular, of order $\sqrt{d}$), then $\|\BW_m^{\top}\BS^{\ell}\|_F$ is of order $\sqrt{nd}(\sqrt{d} + \sqrt{\log n}).$

Now we first start with bounding $\|\BW_m^{\top}\BS^{\ell,m}\|_F$ by using the fact that $\|\BW_m^{\top}\BS^{\ell,m}\|_F^2$ satisfies chi-square distribution.
%{\bf~\cite[Example 2.11]{W19}}
\begin{lemma}\label{lem:chisq}
For $X\sim\chi^2_k$, we have
\[
\Pr( X  \geq k + t)  \leq \exp\left( -\frac{1}{2}\left( t - k\log\left( 1 + \frac{t}{k}\right)\right) \right)
\]
for $t>0.$ In particular, it holds
\[
\Pr(\sqrt{X} \geq  \sqrt{k} + \gamma\sqrt{ \log n}) \leq n^{-\gamma^2/2}, \quad \gamma>0.
\]
\end{lemma} 
\begin{proof}
This is a quite standard result which can be easily proven by using moment generating function and Markov inequality. For any {$\lambda \in (0,1/2)$},
\begin{align*}
\Pr(X\geq k + t) & \leq e^{-\lambda(k+t)}\E e^{\lambda X} = e^{-\lambda(k+t)} (1-2\lambda)^{-k/2} \\
& = \exp\left( -\lambda(k+t) - \frac{k}{2}\log (1-2\lambda)\right).
\end{align*}
The upper bound is minimized when choosing $\lambda = \frac{t}{2(k+t)}$, which gives the bound.

If $t = 2\gamma\sqrt{k\log n} + \gamma^2\log n$, then
\begin{align}
\Pr(X \geq (\sqrt{k} + \gamma\sqrt{\log n})^2) & \leq \exp\left( -\frac{1}{2} \left( 2\gamma\sqrt{k\log n} + \gamma^2\log n - k \log \frac{(\sqrt{k} + \gamma\sqrt{\log n})^2}{k} \right) \right) \nonumber \\
(\text{use}~ \log(1+x) \leq x, x>0))\quad & \leq \exp\left( -\frac{1}{2} \left( 2\gamma\sqrt{k\log n} + \gamma^2\log n - 2k \cdot \frac{\gamma\sqrt{\log n}}{\sqrt{k}} \right) \right) \nonumber  \\
& = n^{-\gamma^2/2} \label{eq:chisqbd}
\end{align}
for any $\gamma > 0.$
\end{proof}

By construction of $\BS^{\ell,m}$, we know that $\BS^{\ell,m}$ is independent of $\BW_m$ for all $\ell$. Thus by taking a union, we have the following lemma.
\begin{lemma}\label{lem:gauss}
It holds uniformly that
\begin{equation}\label{eq:bdchisq}
\max_{1\leq m\leq n}\|\BW_m^{\top}\BS^{\ell,m}\|_F \leq \sqrt{nd}(\sqrt{d} + \gamma\sqrt{\log n})
\end{equation}
for $1\leq \ell\leq L$ with probability at least $1 - L n^{-\gamma^2/2+1}$. In particular, we choose $L = n$ and $\gamma = 4.$
\end{lemma}
\begin{proof}
For any fixed $\ell$ and $m$, we know that each entry of  $\BW_m^{\top}\BS^{\ell,m}\in\RR^{d\times d}$ is an i.i.d. $\mathcal{N}(0,n)$ random variable. Thus
\[
n^{-1}\|\BW_m^{\top}\BS^{\ell,m}\|_F^2\sim \chi^2_{d^2}.
\]
Using Lemma~\ref{lem:chisq} with $k = d^2$ and $t = 2\gamma d\sqrt{\log n}  + \gamma^2\log n$ gives
\[
n^{-1}\|\BW_m^{\top}\BS^{\ell,m}\|_F^2 \leq (d + \gamma\sqrt{\log n})^2
\]
with probability at least $1 - n^{-\gamma^2/2}$. Now taking a union bound over $1\leq m\leq n$ and $1\leq \ell\leq L$, we have
\[
\max_{1\leq \ell\leq L,1\leq m\leq n}\|\BW_m^{\top}\BS^{\ell,m}\|_F \leq  \sqrt{n}\left(d + \gamma\sqrt{\log n}\right) \leq \sqrt{nd}(\sqrt{d} + \gamma\sqrt{\log n})
\]
with probability at least $1 - Ln^{-\gamma^2/2+1} = 1 - n^{-\gamma^2/2+2}.$
\end{proof}
Now we can see that in order to show that $\|\BW_m^{\top}\BS^{\ell}\|_F$ is small, it suffices to bound the error between $\BS^{\ell}$ and $\BS^{\ell,m}$ for all $1\leq m\leq n$ and $1\leq \ell\leq L$. In fact, for the initialization $\BS^0$ and $\BS^{0,m}$, we have the following result which will be proven in Section~\ref{ss:init}.
\begin{lemma}\label{lem:init}
Under 
\[
\sigma < c_0\min\left\{\eps\sqrt{\frac{n}{d}}, \frac{\kappa\sqrt{n}}{\sqrt{d} + 4\sqrt{\log n}}\right\}
\]
for a sufficiently small constant $c_0$, then the initialization satisfies:
\begin{align}
d_F(\BS^0,\BZ) & \leq \frac{\eps\sqrt{nd}}{2}, \label{eq:kappa1-init} \\
d_F(\BS^0,\BS^{0,m}) & \leq \frac{\kappa\sqrt{d}}{2}, \quad 1\leq m\leq n, \label{eq:kappa2-init} \\
\|\BW_m^{\top}\BS^0\|_F & \leq \left(1+\frac{3\kappa}{2} \right)\sqrt{nd}(\sqrt{d} + 4\sqrt{\log n}) \label{eq:kappa3-init}
\end{align}
uniformly for all {$1\leq m\leq n.$}
\end{lemma}
Note that if we only want to prove the tightness of the SDP relaxation, the initialization step is not necessary since we can simply pick $\BS^0 = \BZ$ for the purpose of proof. The initialization step works mainly for the global convergence of the efficient projected generalized power method. 
{We defer the proof of this lemma in Section~\ref{ss:init}. From now on, we set 
\begin{equation}\label{def:xi}
\xi =2+ \frac{3\kappa}{2}, \qquad \kappa = 2 \eps, \qquad \eps\sqrt{d} < \frac{1}{32}
\end{equation}
in $\mathcal{N}_{\xi,\infty}.$ Under this set of parameters, we can see all the lemmas in this section hold under the assumption of Theorem~\ref{thm:sdp} and~\ref{thm:main}.}

\begin{lemma}\label{lem:exp1}
Conditioned on~\eqref{eq:gaussnorm} and~\eqref{eq:bdchisq}, we have for $1\leq \ell\leq L$, 
\begin{align}
d_F(\BS^{\ell}, \BS^{\ell,m}) & \leq \frac{\kappa\sqrt{d}}{2}, \quad 1\leq m\leq n, \label{eq:kappa1}\\
\max_{1\leq m\leq n}\left\| \BW_m^{\top}\BS^{\ell}\right\|_F & \leq 
\left(1+\frac{3\kappa}{2}\right)\sqrt{nd}(\sqrt{d} + 4\sqrt{\log n}), \label{eq:kappa2} \\
d_F(\BS^{\ell}, \BZ) \leq & \frac{\eps\sqrt{nd}}{2},\label{eq:kappa3}
\end{align}
 if
\[
\sigma < \min\left\{ \frac{\eps^2 d}{(3\kappa + 2)\sqrt{d} }, \frac{\kappa}{24 } \right\}\frac{\sqrt{n}}{\sqrt{d} + 4\sqrt{\log n}}, \quad\kappa = 2\eps
\]
which is guaranteed by~\eqref{eq:sigma1} in Theorem~\ref{thm:sdp}.

Combined with Lemma~\ref{lem:con}, we have
\begin{equation}\label{eq:dFcon1}
d_F(\BS^{\ell+1}, \BS^{\ell}) \leq \frac{1}{2} d_F(\BS^{\ell}, \BS^{\ell-1}) \leq 2^{-\ell} d_F(\BS^1, \BS^0), \quad 1\leq \ell\leq L-1
\end{equation}
since $\BS^{\ell}\in \mathcal{N}_{\eps}\cap\mathcal{N}_{\xi,\infty},~\forall 0\leq \ell\leq L.$ 
\end{lemma}

\begin{proof}[\bf Proof of Lemma~\ref{lem:exp1}]
The proof is based on induction. First of all, the initial value $\BS^0$ and $\BS^{0,m}$ satisfy~\eqref{eq:kappa1},~\eqref{eq:kappa2}, and~\eqref{eq:kappa3} in Lemma~\ref{lem:init}.
Suppose all the three inequalities~\eqref{eq:kappa1},~\eqref{eq:kappa2}, and~\eqref{eq:kappa3} hold true for $0,1,\cdots, \ell$, we will prove the case for $\ell+1$ with $\ell\leq L-1.$

{\bf Step 1: Proof of~\eqref{eq:kappa1}.}
Note that $\BS^{\ell}$ is in $\mathcal{N}_{\eps}\cap\mathcal{N}_{\xi,\infty}$ implies $\sigma_{\min}([{\cal L}\BS^{\ell}]_i) > 1-\eps^2 d$ for all $1\leq i\leq n$ which follows from Lemma~\ref{lem:lx}. Applying Lemma~\ref{lem:con} gives
\begin{align*}
d_F(\BS^{\ell+1}, \BS^{\ell+1,m}) 
&  \leq \frac{2}{1-\eps^2d} d_F({\cal L}\BS^{\ell}, {\cal L}^{(m)}\BS^{\ell,m}) \\ & \leq  \frac{9}{4}d_F({\cal L}\BS^{\ell}, {\cal L}^{(m)}\BS^{\ell,m}), \quad 1\leq m\leq n,
\end{align*}
where $\eps\sqrt{ d} \leq 1/32.$
We proceed by using triangle inequality:
\begin{align}
d_F(\BS^{\ell+1}, \BS^{\ell+1,m}) & \leq \frac{9}{4} \left( d_F( {\cal L}\BS^{\ell}, {\cal L}\BS^{\ell,m}) + d_F({\cal L}\BS^{\ell,m}, {\cal L}^{(m)}\BS^{\ell,m})  \right)\nonumber \\
& \leq \frac{9}{4} d_F( {\cal L}\BS^{\ell}, {\cal L}\BS^{\ell,m}) + {\frac{9\sigma}{4n}} \| (\BW - \BW^{(m)}) \BS^{\ell,m}\|_F. \label{eq:slm1}
\end{align}

{Since $d_F(\BS^{\ell}, \BS^{\ell,m}) \leq \kappa\sqrt{d}/2$ and $\kappa = 2\eps$ hold, then 
\[
d_F(\BS^{\ell,m},\BZ) \leq d_F(\BS^{\ell}, \BZ) + d_F(\BS^{\ell}, \BS^{\ell,m}) \leq \frac{\eps\sqrt{nd}}{2} + \frac{\kappa\sqrt{d}}{2} < \eps\sqrt{nd}
\]
implies $\BS^{\ell,m}\in\mathcal{N}_{\eps}.$} 
Then Lemma~\ref{lem:L} implies that
\begin{equation}\label{eq:slm2}
d_F({\cal L}\BS^{\ell}, {\cal L}\BS^{\ell,m}) \leq \frac{1}{9} d_F(\BS^{\ell}, \BS^{\ell,m}).
\end{equation}

Note that { the $j$th block of $(\BW-\BW^{(m)})\BS^{\ell,m}$ satisfies
\[
[(\BW - \BW^{(m)}) \BS^{\ell,m}]_j = 
\begin{cases}
\BW_{jm}\BS_m^{\ell,m}, & j \neq m, \\ 
\BW_m^{\top}\BS^{\ell,m}, & j = m,
\end{cases}
\]
where $\BS_m^{\ell,m}\in\RR^{d\times d}$, the $m$th block of $\BS^{\ell,m}$, is orthogonal.
}
Therefore, we have
\begin{align}
\| (\BW - \BW^{(m)}) \BS^{\ell,m}\|_F & \leq {\|\BW_m \BS^{\ell,m}_m\|_F} + \|\BW_m^{\top}\BS^{\ell,m}\|_F \nonumber \\
(\text{use Lemma~\ref{lem:gauss}}) \quad & \leq {\|\BW_m\|_F} +  \sqrt{nd}(\sqrt{d} + 4\sqrt{\log n})  \nonumber \\
& \leq 3\sqrt{nd} \cdot \sqrt{d} +  \sqrt{nd}(\sqrt{d }+ 4\sqrt{\log n}) \nonumber\\
& < 4\sqrt{nd}(\sqrt{d} + 4\sqrt{\log n}). \label{eq:slm3}
\end{align}

Then combining~\eqref{eq:kappa1},~\eqref{eq:slm1},~\eqref{eq:slm2} and~\eqref{eq:slm3} together gives
\begin{align*}
d_F(\BS^{\ell+1}, \BS^{\ell+1,m}) & \leq \frac{1}{4} d_F(\BS^{\ell}, \BS^{\ell,m}) + \frac{9\sigma}{4n}  \| (\BW - \BW^{(m)}) \BS^{\ell,m}\|_F \\
& \leq \frac{1}{4} \cdot \frac{\kappa\sqrt{d}}{2} + \frac{9\sigma}{4n} \cdot 4\sqrt{nd}(\sqrt{d} + 4\sqrt{\log n}) \leq  \frac{\kappa\sqrt{d}}{2}
\end{align*}
where
$\sigma \leq \frac{\kappa\sqrt{n}}{24(\sqrt{d} + 4\sqrt{\log n}) }.$

{\bf Step 2: Proof of~\eqref{eq:kappa2}.}
%Now we proceed to show that $\|\BW_m^{\top}\BS^{\ell+1}\|_F\leq (3\kappa/2 + 1)\sqrt{nd}(\sqrt{d} + \sqrt{\log n})$. 
Let $\BR$ be the orthogonal matrix which minimizes $\|\BS^{\ell+1} - \BS^{\ell+1,m}\BR\|_F$ over $\BR\in \Od(d)$. Then
\begin{align*}
\left\|\BW_m^{\top}\BS^{\ell+1}\right\|_F & \leq \left\| \BW_m^{\top}(\BS^{\ell+1} - \BS^{\ell+1,m}\BR)\right\|_F + \|\BW_m^{\top}\BS^{\ell+1,m}  \|_F \\
(\text{use Lemma~\ref{lem:gauss}})\quad & \leq \|\BW_m\| \cdot d_F(\BS^{\ell+1}, \BS^{\ell+1,m}) + \sqrt{nd}(\sqrt{d} + 4\sqrt{\log n}) \\
& \leq 3\sqrt{nd} \cdot \frac{\kappa\sqrt{d}}{2} +  \sqrt{nd}(\sqrt{d} + 4\sqrt{\log n}) \\
 & \leq \left(1+\frac{3\kappa}{2} \right)\sqrt{nd}(\sqrt{d} + 4\sqrt{\log n}).
\end{align*}

{\bf Step 3: Proof of~\eqref{eq:kappa3}.}
%\paragraph{Proof of~\eqref{eq:kappa3}}
Since both $\BS^{\ell}$ and $\BZ\in\mathcal{N}_{\eps}\cap\mathcal{N}_{\xi,\infty}$, then
\begin{align*}
d_F(\BS^{\ell+1}, \BZ) & = d_F(\PP_n( {\cal L} \BS^{\ell}), \PP_n( \BZ))\leq \frac{2}{1-\eps^2 d} \cdot d_F({\cal L}\BS^{\ell}, \BZ) \\
&  \leq \frac{9}{2} \cdot d_F({\cal L}\BS^{\ell}, \BZ), \qquad\BS^{\ell}\in\mathcal{N}_{\eps}\cap \mathcal{N}_{\xi,\infty}
\end{align*}
which follows from Lemma~\ref{lem:con}.

For $d_F({\cal L}\BS^{\ell}, \BZ)$, it suffices to estimate $\max_{1\leq i\leq n}\| [{\cal L}\BS^{\ell}]_i -\BR \|_F$ for some orthogonal matrix $\BR = \PP(\BZ^{\top}\BS^{\ell})\in \Od(d).$ We look at the $i$th block of ${\cal L}\BS^{\ell}:$
\[
[{\cal L}\BS^{\ell}]_i  = \frac{1}{n}\BZ^{\top}\BS^{\ell} + \frac{\sigma}{n} \BW_i^{\top}\BS^{\ell} 
\]
Then
\begin{align*}
\min_{\BR\in \Od(d)}\max_{1\leq i\leq n} \|[{\cal L}\BS^{\ell}]_i - \BR \|_F & \leq \min_{\BR\in \Od(d)} \left\| \frac{1}{n}\BZ^{\top}\BS^{\ell} - \BR \right\|_F + \frac{\sigma}{n}\max_{1\leq i\leq n} \|\BW_i^{\top}\BS^{\ell}\|_F.
\end{align*}
For the first term, we have $\min_{\BR\in \Od(d)}\|\BS^{\ell} - \BZ\BR\|_F\leq \eps\sqrt{nd}/2$ which is equivalent to
\[
2nd - 2\|\BZ^{\top}\BS^{\ell}\|_* \leq \frac{\eps^2 nd}{4} \Longrightarrow \|\BZ^{\top}\BS^{\ell}\|_* \geq \left(1-\frac{\eps^2}{8}\right)nd.
\]
Note that
\begin{align*}
\|\BZ^{\top}\BS^{\ell} - n\BR\|_F^2 & = \sum_{i=1}^d \left(\sigma_i^2(\BZ^{\top}\BS^{\ell}) - 2n\sigma_i(\BZ^{\top}\BS^{\ell})\right) + n^2d \\
& = \sum_{i=1}^d (\sigma_i(\BZ^{\top}\BS^{\ell}) - n)^2\\
& \leq  (nd - \|\BZ^{\top}\BS^{\ell}\|_*)^2 \\
& \leq \frac{\eps^4 n^2d^2}{64} 
\end{align*}
where $0\leq \sigma_{i}(\BZ^{\top}\BS^{\ell})\leq n.$ This implies
$\|\BZ^{\top}\BS^{\ell} - n\BR\|_F \leq \eps^2 nd/8.$
Then
\begin{align*}
\min_{\BR\in \Od(d)}\max_{1\leq i\leq n} \|[{\cal L}\BS^{\ell}]_i - \BR \|_F & \leq \min_{\BR\in \Od(d)} \left\| \frac{1}{n}\BZ^{\top}\BS^{\ell} - \BR \right\|_F + \frac{\sigma}{n}\max_{1\leq i\leq n} \|\BW_i^{\top}\BS^{\ell}\|_F \\
(\text{use}~\eqref{eq:kappa2})\quad & \leq \frac{\eps^2d}{8}  + \frac{\sigma}{n}\cdot \left(1+\frac{3\kappa}{2}\right) \sqrt{nd}(\sqrt{d} + 4\sqrt{\log n}) \\
& \leq \frac{\eps^2 d}{8} + \frac{\eps^2 d}{2}  < \frac{\eps\sqrt{d}}{9}
\end{align*}
where $\eps$ is a constant satisfying $\eps\sqrt{d} \leq 1/32$ and
\[
\sigma \leq \frac{ \eps^2  d}{3\kappa+2} \cdot \frac{ \sqrt{n}}{\sqrt{d}(\sqrt{d} + 4\sqrt{\log n})}.
\]

Therefore, it holds
\[
d_F(\BS^{\ell+1},\BZ)\leq \frac{9}{2}\cdot d_F({\cal L}\BS^{\ell}, \BZ) \leq \frac{9}{2} \cdot \frac{\eps\sqrt{nd}}{9}\leq \frac{\eps\sqrt{nd}}{2}.
\]
Now we have proven $\BS^{\ell+1}\in\mathcal{N}_{\eps}\cap\mathcal{N}_{\xi,\infty}$. By induction, we have $\BS^{\ell}\in\mathcal{N}_{\eps}\cap\mathcal{N}_{\xi,\infty} $ for all $1\leq \ell\leq L$.
Now we can immediately invoke Lemma~\ref{lem:con} and get
\begin{align}
d_F(\BS^{\ell+1}, \BS^{\ell}) & = d_F(\PP_n({\cal L}\BS^{\ell}), \PP_n({\cal L}\BS^{\ell-1})) \leq \frac{1}{2} d_F(\BS^{\ell}, \BS^{\ell-1}). \nonumber
\end{align}
\end{proof}

\begin{lemma}\label{lem:exp2}
Conditioned on the Lemma~\ref{lem:exp1}, we have
\begin{align}\label{eq:df_exp}
d_F(\BS^{L+t}, \BS^{L+t-1}) & \leq 2^{-t} d_F(\BS^{L},\BS^{L-1}),\quad \forall t\in\mathbb{Z}_{\geq 0}
\end{align}
and $\{\BS^{L+t}\}_{t\geq 0}$ belong to $\mathcal{N}_{\eps} \cap\mathcal{N}_{\xi,\infty}$, i.e., 
\begin{align}
\max_{1\leq m\leq n}\left\|\BW_m^{\top}\BS^{\ell}\right\|_F & \leq \left(2 +\frac{3\kappa}{2}\right) \sqrt{nd}(\sqrt{d} + 4\sqrt{\log n}), \label{eq:kappa2_v2} \\
d_F(\BS^{\ell}, \BZ) & \leq  \eps\sqrt{nd},\label{eq:kappa3_v2}
\end{align}
for all $\ell\in\mathbb{Z}_{\geq 0}.$
Given~\eqref{eq:dFcon1} in Lemma~\ref{lem:exp1}, we have
\[
d_F(\BS^{\ell}, \BS^{\ell-1})  \leq 2^{-\ell+1} { d_F(\BS^1, \BS^0)}, \quad \forall \ell\in\mathbb{Z}_{\geq 0}.
\]
Here we set $L = n.$
\end{lemma}
\begin{proof}

We will prove~\eqref{eq:df_exp}
for all $t\in\mathbb{Z}_{\geq 0}$ by induction. This is true if $t=0$. Now we assume that this also holds for all $0\leq t \leq k$ and prove the case for $t = k+1.$ 
In fact, it suffices to show that $\BS^{L+k}\in\mathcal{N}_{\eps}\cap\mathcal{N}_{\xi,\infty}$. If so, we have 
\[
d_F(\BS^{L+k+1}, \BS^{L+k}) = d_F(\PP_n({\cal L}\BS^{L+k}), \PP_n({\cal L}\BS^{L+k-1})) \leq\frac{1}{2} d_F(\BS^{L+k}, \BS^{L+k-1}),
\]
following from Lemma~\ref{lem:con}.

Now we will start proving $\BS^{L+k}\in\mathcal{N}_{\eps}\cap\mathcal{N}_{\xi,\infty}$. Combining Lemma~\ref{lem:exp1} with the assumption indicates that
\[
\BS^L\in\mathcal{N}_{\eps}\cap\mathcal{N}_{\xi,\infty}, \quad d_F(\BS^{\ell}, \BS^{\ell-1}) \leq 2^{-\ell+1}d_F(\BS^1, \BS^0), \quad 1\leq \ell\leq L+k.
\]
By triangle inequality and the assumption that~\eqref{eq:df_exp} holds for $0\leq t\leq k$, then{
\begin{align*}
d_F(\BS^{L+k}, \BS^L) & \leq \sum_{t=1}^k d_F(\BS^{L+t}, \BS^{L+t-1})  \leq \sum_{t=1}^k 2^{-t}d_F(\BS^L, \BS^{L-1}) \\
(\text{Use}~\eqref{eq:dFcon1})\quad & \leq d_F(\BS^L, \BS^{L-1})  \leq 2^{-L+1}d_F(\BS^1, \BS^0) \\
&  \leq 2^{-L+1}\cdot \eps\sqrt{nd} \\
(\text{Use}~\eps\sqrt{d}\leq 1/32)\quad &  \leq 2^{-n+1}\sqrt{n}\cdot \frac{1}{32} < \frac{1}{3}
\end{align*}
where $d_F(\BS^1, \BS^0)\leq d_F(\BS^1,\BZ) + d_F(\BS^0,\BZ)\leq \eps\sqrt{nd}$ follows from~\eqref{eq:kappa3} and  $L$ is $n$.}
As a result, we have{
\begin{align*}
d_F(\BS^{L+k}, \BZ) & \leq d_F(\BS^{L+k}, \BS^L) + d_F(\BS^L, \BZ) \leq 2^{-n+1}\cdot\eps\sqrt{nd}+ \frac{\eps\sqrt{nd}}{2} \leq \eps\sqrt{nd}
\end{align*}
which implies $\BS^{L+k}\in\mathcal{N}_{\eps}$}. Denote $\BR$ as the orthogonal matrix which minimizes $\|\BS^{L+k} - \BS^L\BR\|_F$.
For $\|\BW_m^{\top}\BS^{L+k}\|_F$, we have
\begin{align*}
\|  \BW_m^{\top}\BS^{L+k} \|_F & \leq \|\BW_m^{\top}(\BS^{L+k} - \BS^{L}\BR)\|_F + \|\BW_m^{\top}\BS^L\|_F\\
(\text{use}~\eqref{eq:kappa2})\quad & \leq \|\BW_m\| \cdot d_F(\BS^{L+k}, \BS^L) + \left(1+\frac{3\kappa}{2} \right)\sqrt{nd} (\sqrt{d} + 4\sqrt{\log n}) \\
(d_F(\BS^{L+k}, \BS^L)\leq 1/3)\quad & \leq {3\sqrt{nd} \cdot \frac{1}{3}}+ \left(1+\frac{3\kappa}{2} \right) \sqrt{nd} (\sqrt{d} + 4\sqrt{\log n}) \\
& \leq \left(2+\frac{3\kappa}{2}\right)\sqrt{nd}(\sqrt{d} + 4\sqrt{\log n}) \\
& \leq \xi \sqrt{nd}(\sqrt{d} + 4\sqrt{\log n})
\end{align*}
where $\|\BW_m\|\leq 3\sqrt{nd}$, the bound on $\|\BW_m^{\top}\BS^L\|_F$ follows from~\eqref{eq:kappa2}, and $\xi = 3\kappa/2+2$ with $\kappa=2\eps$.

Now we have shown that $\BS^{L+k}$ and $\BS^{L+k-1}$ are in the contraction region $\mathcal{N}_{\eps}\cap\mathcal{N}_{\xi,\infty}$. Applying Lemma~\ref{lem:con} finishes the proof:
\begin{align*}
d_F(\BS^{L+k+1}, \BS^{L+k}) & = d_F (\PP_n({\cal L} \BS^{L+k}), \PP_n ({\cal L} \BS^{L+k-1})) \leq \frac{1}{2}d_F(\BS^{L+k}, \BS^{L+k-1}).
\end{align*}
\end{proof}

\subsection{Proof of Theorem~\ref{thm:sdp} and~\ref{thm:main}: Convergence to global maximizer}
\label{ss:main_proof}
\begin{proof}[\bf Proof of Theorem~\ref{thm:sdp} and~\ref{thm:main}]
The proofs of Theorem~\ref{thm:sdp} and~\ref{thm:main} are exactly the same except the initialization part. For the tightness of SDP, it suffices to let $\BS^0 = \BZ$ since we only need to show the global optimum of the SDP is exactly the MLE. For the generalized power method, we show that spectral initialization provides a sufficiently good starting point which is addressed in Section~\ref{ss:init}.

To prove Theorem~\ref{thm:sdp} and~\ref{thm:main}, it suffices to show that the iterates have a limiting point $\BS^{\infty}$ which satisfy $\BA\BS^{\infty} = \BLambda\BS^{\infty}$ and is also located in $\mathcal{N}_{\eps}\cap\mathcal{N}_{\xi,\infty}$, according to Theorem~\ref{thm:optlimit}.

{\bf Step 1:} The existence and convergence of the limiting point in $\mathcal{N}_{\eps}\cap\mathcal{N}_{\xi,\infty}$. 

Conditioned on Lemma~\ref{lem:exp2}, it holds that $d_F(\BS^{\ell}, \BS^{\ell-1}) \leq 2^{-\ell+1}d_F(\BS^1,\BS^0)$ for all $\ell\in\mathbb{Z}_{\geq 0}.$ 
Then $\{\BS^{\ell}\}_{\ell=0}^{\infty}$ is a Cauchy sequence since 
\begin{align}
d_F(\BS^{\ell+k}, \BS^{\ell}) & \leq \sum_{t=\ell}^{\ell+k-1} d_F(\BS^{t+1}, \BS^{t})\nonumber \\
& \leq \sum_{t=\ell}^{\ell+k-1} 2^{-t}d_F(\BS^1,\BS^0) \leq 2^{-\ell+1} d_F(\BS^1,\BS^0), \quad \forall k\in\mathbb{Z}_{\geq 0} \label{eq:cauchy}
\end{align}
is arbitrarily small for any sufficiently large $\ell.$
Note that $\mathcal{N}_{\eps}\cap\mathcal{N}_{\xi,\infty}$ is a closed and compact set in the Euclidean space. 
The distance function $d_F(\cdot,\cdot)$ is equivalent to the Frobenius norm on the quotient set $\Od(d)^{\otimes n}/\Od(d)$ where $\Od(d)$ is also a compact set. 
 Thus since $\BS^{\ell}$ is a Cauchy sequence, we know there must exist a limiting point $\BS^{\infty}.$ From~\eqref{eq:cauchy}, we also have the linear convergence of $\BS^{\ell}$ to the limiting point:
 \begin{equation}\label{eq:linear}
 d_F(\BS^{\infty}, \BS^{\ell}) \leq 2^{-\ell+1}d_F(\BS^1, \BS_0).
 \end{equation}
Now we want to prove that $\BS^{\infty}(\BS^{\infty})^{\top}$ is the unique global maximizer to~\eqref{def:sdp} by using Theorem~\ref{thm:optlimit}.

Note that~\eqref{eq:kappa2_v2} and~\eqref{eq:kappa3_v2} hold for all $\ell$. Therefore, sending $\ell$ to $\infty$ gives 
\begin{equation}\label{eq:Sinfty}
\max_{1\leq m\leq n}\|\BW_m^{\top}\BS^{\infty}\|_F \leq \left( 2+\frac{3\kappa}{2}\right)\sqrt{nd}(\sqrt{d} + \sqrt{\log n}),\quad
d_F(\BS^{\infty}, \BZ) \leq \eps\sqrt{nd}
\end{equation}
since $\mathcal{N}_{\eps}$ and $\mathcal{N}_{\xi,\infty}$ are compact sets. This implies $\BS^{\infty}\in\mathcal{N}_{\eps}\cap\mathcal{N}_{\xi,\infty}.$

{\bf Step 2:} The limiting point $\BS^{\infty}$ satisfies $\BA\BS^{\infty} = \BLambda\BS^{\infty}$ for some block-diagonal positive semidefinite matrix $\BLambda.$

Note that
\[
d_F({\cal P}_n({\cal L}\BS^{\infty}), \BS^{\ell}) \leq d_F({\cal P}_n({\cal L}\BS^{\infty}), {\cal P}_n({\cal L}\BS^{\ell})) + d_F({\cal P}_n({\cal L}\BS^{\ell}), \BS^{\ell}).
\]
By letting $\ell$ go to $\infty$ and using Lemma~\ref{lem:con}, it holds
\[
\lim_{\ell\rightarrow\infty}d_F({\cal P}_n({\cal L}\BS^{\infty}), {\cal P}_n({\cal L}\BS^{\ell})) \leq \frac{1}{2}\lim_{\ell\rightarrow\infty} d_F(\BS^{\infty}, \BS^{\ell}) = 0.
\]
For $d_F({\cal P}_n({\cal L}\BS^{\ell}), \BS^{\ell})$, we have
\[
d_F({\cal P}_n({\cal L}\BS^{\ell}), \BS^{\ell}) = d_F(\BS^{\ell+1}, \BS^{\ell}) \leq 2^{-\ell}d_F(\BS^1, \BS^0) \rightarrow 0, \quad \ell\rightarrow\infty,
\]
which follows from Lemma~\ref{lem:exp1}.
Therefore, 
$d_F({\cal P}_n({\cal L}\BS^{\infty}), \BS^{\infty}) = 0$ 
and we have ${\cal P}_n({\cal L}\BS^{\infty}) = \BS^{\infty}\BO$ for some orthogonal matrix $\BO.$

Note that $d_F(\BS^{\infty}, \BZ) \leq \eps\sqrt{nd}$ and 
$\left\|\BW_i^{\top}\BS^{\infty}\right\|_F \leq  \left(2+3\kappa/2\right)\sqrt{nd}(\sqrt{d} + \sqrt{\log n})$ imply $\sigma_{\min}([{\cal L}\BS^{\infty}]_i) > 0$.
Therefore, the limiting point $d_F({\cal P}_n({\cal L}\BS^{\infty}), \BS^{\infty}) = 0$ satisfies: 
\[
\PP([ {\cal L}\BS^{\infty} ]_i ) = ([ {\cal L}\BS^{\infty} ]_i [ {\cal L}\BS^{\infty} ]_i^{\top})^{-\frac{1}{2}} [ {\cal L}\BS^{\infty} ]_i= \BS_i^{\infty} \BO, \quad \forall 1\leq i\leq n
\]
where $\BO$ is an orthogonal matrix. It means
\begin{equation}\label{eq:1st}
 [ {\cal L}\BS^{\infty} ]_i = ([ {\cal L}\BS^{\infty} ]_i [ {\cal L}\BS^{\infty} ]_i^{\top})^{\frac{1}{2}}\BS^{\infty}_i \BO.
\end{equation}
Next, we will show that $\BO = \I_d.$

From~\eqref{eq:1st}, we know that
\[
\sum_{i=1}^n (\BS^{\infty}_i)^{\top}[{\cal L}\BS^{\infty}]_i = \sum_{i=1}^n(\BS^{\infty}_i)^{\top} ([ {\cal L}\BS^{\infty} ]_i [ {\cal L}\BS^{\infty} ]_i^{\top})^{\frac{1}{2}}\BS^{\infty}_i \BO
\]
For the term on the left side,
\begin{align*}
\sum_{i=1}^n (\BS^{\infty}_i)^{\top}[{\cal L}\BS^{\infty}]_i & =\frac{1}{n} \sum_{i,j} (\BS^{\infty}_i)^{\top}(\I + \sigma\BW_{ij})\BS^{\infty}_j \\
& = \frac{1}{n}(\BS^{\infty})^{\top}\BZ\BZ^{\top}\BS^{\infty} + \frac{\sigma}{n} (\BS^{\infty})^{\top}\BW\BS^{\infty} \\
& \succeq (n^{-1}\sigma_{\min}^2(\BZ^{\top}\BS^{\infty}) - \sigma \|\BW\|)\I_d \\
& \geq \left(n \left(1 - \frac{\eps^2 d}{2}\right)^2 - 3\sigma \sqrt{nd} \right)\I_d \succ 0
\end{align*}
where
\[
\sigma_{\min}(\BZ^{\top}\BS^{\infty}) \geq \left(1-\frac{\eps^2d}{2}\right)n, \quad \sigma < \frac{1}{3}\left(1 - \frac{\eps^2 d}{2}\right)^2 \sqrt{\frac{n}{d}}.
\]

Now we let
\[
\BX = \sum_{i=1}^n (\BS^{\infty}_i)^{\top}[{\cal L}\BS^{\infty}]_i  \succ 0, \quad  \BY= \sum_{i=1}^n(\BS^{\infty}_i)^{\top} ([ {\cal L}\BS^{\infty} ]_i [ {\cal L}\BS^{\infty} ]_i^{\top})^{\frac{1}{2}}\BS^{\infty}_i \succeq 0.
\]
For $\BX\succ 0$ and $\BY\succeq 0$ which satisfy 
$\BX =\BY\BO$, we will show $\BO$ must be $\I_d.$
This follows from
\[
\BX^{-\frac{1}{2}} \BY \BX^{-\frac{1}{2}} = \BX^{\frac{1}{2}} \BO^{\top} \BX^{-\frac{1}{2}}
\]
The left hand side is still positive semidefinite and all its eigenvalues are real and nonnegative. The right hand side has eigenvalues which are in the form of $\{\exp(\mi \theta_i)\}_{i=1}^n$ since $\BO$ is orthogonal. Thus $\theta_i$ must be zero and its means $\BO = \I_d.$

Therefore, following from~\eqref{eq:1st}, the limiting point satisfies
\[
\BA\BS^{\infty} = \BLambda \BS^{\infty}
\]
where $\BLambda_{ii} = n([ {\cal L}\BS^{\infty} ]_i [ {\cal L}\BS^{\infty} ]_i^{\top})^{\frac{1}{2}} \succ 0.$ Moreover, we have shown that $\BS^{\infty}\in \mathcal{N}_{\eps}\cap\mathcal{N}_{\xi,\infty}$. Applying Theorem~\ref{thm:optlimit} immediately gives the global optimality of $\BS^{\infty}(\BS^{\infty})^{\top}$ in the SDP relaxation~\eqref{def:sdp}. 

{\bf Step 3: } A block-wise error bound between $\BS^{\infty}$ and $\BZ$.

Now we will derive a block-wise error bound of $\BS^{\infty}$,  
\[
\min_{\BO\in\Od(d)}\max_{1\leq i\leq n}\| \BS^{\infty}_i - \BO\|_F.
\]
This is bounded by $\max_{1\leq i\leq n}\| \BS^{\infty}_i - \BQ\|_F$ where $\BQ = \PP(\BZ^{\top}\BS^{\infty}).$
Note that $\BS^{\infty}$ satisfies~\eqref{eq:Sinfty} and Lemma~\ref{lem:sigmamin} implies $\sigma_{\min} \left( \BZ^{\top} \BS^{\infty}\right) \geq \left( 1-\eps^2d/2\right)n$. For each $1\leq i\leq n$, we have $\BS_i^{\infty} = \PP( [{\cal L}\BS^{\infty}]_i )$ and
\begin{align*}
\| \BS_i^{\infty} - \BQ \|_F & = \|  \PP( [{\cal L}\BS^{\infty}]_i ) - \PP(n^{-1}\BZ^{\top}\BS^{\infty})\|_F \\
& \leq \frac{2}{1-\eps^2 d} \| [{\cal L}\BS^{\infty}]_i - n^{-1}\BZ^{\top}\BS^{\infty}\|_F \\
& = \frac{2}{1-\eps^2 d} \left\| \frac{\sigma}{n}\BW_i^{\top}\BS^{\infty} \right\|_F \\
(\text{use~\eqref{eq:Sinfty}})\quad & \leq  \frac{2}{1-\eps^2 d} \cdot  \frac{\sigma}{n} \cdot \left(2 + \frac{3\kappa}{2}\right)\sqrt{nd}(\sqrt{d} + 4\sqrt{\log n}) \\
& \lesssim \sigma\sqrt{n^{-1}d}(\sqrt{d} + 4\sqrt{\log n})
\end{align*}
which follows from Lemma~\ref{lem:phase} and~\eqref{eq:Sinfty}.
\end{proof}

\subsection{Proof of Lemma~\ref{lem:init}: Initialization}\label{ss:init}

Recall that the initialization in Algorithm 1 involves two steps: (a) compute the top $d$ eigenvectors $\BPhi$ of $\BA = \BZ\BZ^{\top} + \sigma
\BW$; (b) round each $d\times d$ block $\BPhi_i$ of $\BPhi$ into an orthogonal matrix and use $\BS_i^0 = \PP(\BPhi_i)$ as the initialization. Note that Lemma~\ref{lem:init} requires the initial value $\BS^{0}$ to be very close to the ground true $\BZ$ and also not highly ``aligned" with the noise $\BW_m$. It is a seemingly  simple task to justify that:
if the noise is small, then applying Davis-Kahan theorem would provide a tight bound of $\min_{\BQ\in\Od(d)}\|\BPhi - \BZ\BQ\|$. Then we believe that $\PP_n(\BPhi)$ is also close to $\BZ$. 
However, due to the rounding procedure in the initialization, we actually need a more careful analysis of the eigenvectors $\BPhi$. 
%More precisely, the rounding procedure equals
%\[
%\BS_i^0 = \PP(\BPhi_i), \quad 1\leq i\leq n.  %= (\BPhi_i\BPhi_i^{\top})^{-\frac{1}{2}}\BPhi_i
%\]
%which involves a block-wise inversion. ({\color{red} Need to specify why $\BPhi_i$ is invertible.})

In order to obtain $d_F(\BS^0,\BZ) \leq\eps\sqrt{nd}$, we actually need to establish a block-wise bound between $\BPhi$ and $\BZ$, i.e., 
\[
\min_{\BQ\in\Od(d)}\max_{1\leq i\leq n} \|\BPhi_i - \BQ\|.
\]
As a result, what we need is much more than classical eigenspace perturbation bound given by Davis-Kahan theorem~\cite{DK70} which provides an error bound in $\|\cdot\|$ or $\|\cdot\|_F$. 

The block-wise perturbation bound on the eigenvectors is a natural generalization of recent entrywise error bound on eigenvectors arising from spectral clustering, $\mathbb{Z}_2$-synchronization, and matrix completion. The author's early work~\cite{L20b} uses the popular leave-one-out technique to obtain such a block-wise error bound in several applications, including the scenario discussed in this paper. One core result for the initialization reads as follows:
\begin{theorem}[Theorem 5.1 in~\cite{L20b}]\label{thm:loo}
Suppose $(\BPhi,\BLambda)$ are the top $d$ eigenvectors and eigenvalues of $\BA$ satisfying $\BA\BPhi = \BPhi\BLambda$, $\BPhi^{\top}\BPhi = n\I_d$ and $\BQ= \PP(\BZ^{\top}\BPhi)$. Let
\begin{equation}\label{def:etaod}
\eta := \sigma n^{-\frac{1}{2}}(\sqrt{d} + \sqrt{\log n}) , \quad \sigma< C_0^{-1} \sqrt{n}(\sqrt{d} + \sqrt{\log n})^{-1}
\end{equation}
where $C_0 > 0$ is an absolute large constant.
Then with probability at least $1 - O(n^{-1}d^{-1})$, it holds 
\begin{equation}\label{eq:key-init}
\|\BPhi_i - [\BA\BZ]_i\BQ\BLambda^{-1}\|  \lesssim \eta \max_{1\leq i\leq n}\|\BPhi_i\|,
\end{equation}
uniformly for all $1\leq i\leq n$ where $[\BA\BZ]_i$ is the $i$th block of $\BA\BZ$. Moreover, we have
\begin{equation}\label{eq:phi_sv}
| \sigma_j(\BPhi_i) - 1| \lesssim \eta,\qquad \forall 1\leq i\leq n,~1\leq j\leq d.
\end{equation}
\end{theorem}

Interested readers may refer to~\cite{L20b} for more details and we do not repeat the proof again here. For the completeness of presentation, we will discuss the main idea of proving~\eqref{eq:key-init} and how to derive~\eqref{eq:phi_sv} from~\eqref{eq:key-init}. Note that our goal is to estimate each block of $\BPhi$. The idea is to apply one step power iteration and use the outcome to approximate $\BPhi$:
\[
\BPhi - \BA \BZ \BQ \BLambda^{-1} = \BA ( \BPhi - \BZ \BQ )\BLambda^{-1} 
\]
where ${\BA\BPhi = \BPhi\BLambda}.$ Therefore, the error bound on the $i$th block is
\[
\|\BPhi_i - [\BA \BZ]_i \BQ \BLambda^{-1}\| \leq \|[\BA(\BPhi - \BZ\BQ)]_i\|\|\BLambda^{-1}\|
\]
Here $\lambda_{\min}(\BLambda) \geq n - \sigma\|\BW\| > n/2$ follows from Weyl's theorem and thus $\|\BLambda^{-1}\|\leq 2/n.$
It suffices to bound $\|[\BA(\BPhi - \BZ\BQ)]_i\|$. However, the $i$th block column of $\BA$ and $\BPhi$ are~\emph{statistically dependent} and thus we cannot apply concentration inequality directly to obtain a tight bound. To overcome this issue, we apply the leave-one-out technique again: approximate $\BPhi$ with $\BPhi^{(i)}$ which consists of the top $d$ eigenvectors of $\BA^{(i)}$ defined in~\eqref{eq:Am}: 
\begin{align*}
\|\BA_i^{\top}(\BPhi - \BZ\BQ)\| & = \|(\BZ+\sigma\BW_i)^{\top}(\BPhi - \BZ\BQ) \| \\
& \leq \|\BZ^{\top}(\BPhi - \BZ\BQ) \| + \sigma\|\BW_i^{\top}(\BPhi -\BPhi^{(i)} ) \| + \sigma\|\BW_i^{\top}(\BPhi^{(i)}- \BZ\BQ) \| \\
& \leq \sqrt{n}\|\BPhi - \BZ\BQ\| + \sigma \|\BW_i\|\|\BPhi - \BPhi^{(i)}\| + \sigma \|\BW_i^{\top}(\BPhi^{(i)} - \BZ\BQ)\|.
\end{align*}
The first term is bounded by using standard Davis-Kahan argument; for the second term, we first use Davis-Kahan theorem again to bound $\|\BPhi - \BPhi^{(i)}\|$ which is quite small and then control the term by $\|\BW_i\|\|\BPhi - \BPhi^{(i)}\|$; for the third term, due to the independence between $\BW_i$ and $\BPhi^{(i)} -\BZ\BQ$, one can also get a tight bound by using concentration inequality of a Gaussian random matrix~\cite{V18}. This finishes the proof of~\eqref{eq:key-init}. 

With~\eqref{eq:key-init}, we can derive~\eqref{eq:phi_sv}. First of all, $\max_{1\leq i\leq n}\|\BPhi_i\| = O(1)$ holds since~\eqref{eq:key-init} gives
\[
\max_{1\leq i\leq n}\|\BPhi_i\| \lesssim \frac{1}{1-C'\eta}\| [\BA\BZ]_i\BQ\BLambda^{-1} \| \leq C_0 \|[\BA\BZ]_i\|\|\BLambda^{-1}\| \leq 4C_0.
\]
Here $C_0$ and $C'$ are two constants, and the bound on $\|[\BA\BZ]_i\|$ follows from
\[
\|[\BA\BZ]_i\| =\left\|\sum_{j=1}^n (\I_d + \sigma\BW_{ij})\right\| \leq n + \sigma \|\BW_i^{\top}\BZ\| \leq 2n
\]
if $\sigma < c_0\sqrt{n}/(\sqrt{d} +\sqrt{\log n})$ where $\|\BW_i^{\top}\BZ\|=O( \sqrt{n}(\sqrt{d} + \sqrt{\log n}))$ holds with high probability for all $1\leq i\leq n$ implied by~\cite[Theorem 4.4.5]{V18}.  Now we have
\begin{align*}
\|\BPhi_i - \BQ\| & \leq \|\BPhi_i - [\BA\BZ]_i\BQ\BLambda^{-1}\| + \|[\BA\BZ]_i\BQ\BLambda^{-1} - \BQ\| \\
\text{(use~\eqref{eq:key-init})}~~  & \lesssim \eta + \| (n\I_d + \BW_i^{\top}\BZ)\BQ\BLambda^{-1} - \BQ \|\\
& \lesssim \eta + \sigma\| \BW_i^{\top}\BZ\BQ\BLambda^{-1} \| + \| \BQ (n\BLambda^{-1} - \I_n) \| \\
& \lesssim \eta + \sigma \|\BW_i^{\top}\BZ\| \|\BLambda^{-1}\| + \|\BLambda - n\I_d\| \|\BLambda^{-1}\| \\
& \lesssim \eta + \sigma \sqrt{n}(\sqrt{d} + \sqrt{\log n}) \cdot \frac{2}{n} + \sigma\|\BW\| \cdot \frac{2}{n}  \lesssim 6\eta
\end{align*}
where $\|\BLambda - n\I_d\| \leq \sigma\|\BW\|$.
These arguments lead to~\eqref{eq:phi_sv} in Lemma~\ref{thm:loo}.

The next lemma is used to estimate the distance $d_F(\BS^{0}, \BS^{0,m})$ in~\eqref{eq:kappa2-init}. In fact, it suffices to estimate the distance between $\BPhi$ and $\BPhi^{(m)}$, i.e., the top $d$ eigenvectors of $\BA$ and $\BA^{(m)}$ respectively. Note that $\BA$ and $\BA^{(m)}$ differ only by one block column and row, and thus their eigenvectors are also very close, following from Davis-Kahan theorem. 
\begin{lemma}[Lemma 5.9 in~\cite{L20b}]\label{lem:dk2}
Let $\BPhi$ and $\BPhi^{(m)}$ be the top $d$ eigenvectors of $\BA$ and $\BA^{(m)}$ with $\BPhi^{\top}\BPhi  = (\BPhi^{(m)})^{\top}\BPhi^{(m)}  = n\I_d$ respectively. Then 
\begin{align*}
\|\BPhi - \BPhi^{(m)}\BQ_m \| \leq \sqrt{2}\| ( \I - n^{-1} \BPhi \BPhi^{\top} )\BPhi^{(m)}\| \lesssim\sigma  n^{-\frac{1}{2}}(\sqrt{d} + \sqrt{\log n}) \cdot \max_{1\leq i\leq n} \|\BPhi_i\|.
\end{align*}
where $\BQ_m = \PP((\BPhi^{(m)})^{\top}\BPhi)$.
In other words, 
\[
n - \sigma_{\ell}\left( \BPhi^{\top} \BPhi^{(m)} \right) \lesssim  \sigma (\sqrt{d} + \sqrt{\log n})\max_{1\leq i\leq n} \|\BPhi_i\|, \quad 1\leq \ell\leq d
\]
and
\[
\| (\BPhi^{\top} \BPhi^{(m)} (\BPhi^{(m)})^{\top}\BPhi)^{1/2} - n\I_d \| \lesssim \sigma(\sqrt{d} + \sqrt{\log n}) \max_{1\leq i\leq n} \|\BPhi_i\|.
\]

\end{lemma}
\begin{proof}
The proof follows from setting $\BX = \BA$ and $\BX_{\BE} = \BA^{(m)}$ and then applying Davis-Kahan perturbation bound in Theorem~\ref{thm:DK}. 
\end{proof}

Now we will show that the initialization step produces $\BS^0$ which satisfies~\eqref{eq:kappa1-init},~\eqref{eq:kappa2-init}, and~\eqref{eq:kappa3-init}.

\begin{proof}[\bf Proof of Lemma~\ref{lem:init}]
By the Davis-Kahan theorem with $\BX = \BZ\BZ^{\top}$ and $\BX_{\BE} = \BX + \sigma\BW$, it holds that
\begin{align*}
\|\BPhi - \BZ\BQ\|_F & \leq \sqrt{2}\left\| \left( \I_{nd} - \frac{1}{n} \BPhi  \BPhi ^{\top} \right)\BZ\right\|_F \\
& \leq \sqrt{2}\cdot \frac{1}{n - \sigma\|\BW\|}\cdot \| (\BA - \BZ\BZ^{\top})\BZ \|_F \\
&  \lesssim \sigma n^{-1} \cdot \| \BW\BZ \|_F \\
& \lesssim \sigma n^{-1}\cdot 3\sqrt{nd}\cdot \sqrt{nd} \lesssim \sigma d
\end{align*}
where $\|\BZ\|_F = \sqrt{nd}$ and $\BQ = \PP(\BZ^{\top}\BPhi)$.
Note that all the blocks of $\BPhi$ have smallest singular values greater than 1/2,  following from Lemma~\ref{thm:loo}. Thus
\begin{align*}
d_F(\BS^0,\BZ) & \leq \| \PP_n(\BPhi) - \PP_n(\BZ\BQ) \|_F \\
& \leq \sqrt{\sum_{i=1}^n \|\PP(\BPhi_i) - \PP(\BQ)\|_F^2} \\
& \leq 2\sqrt{ \sum_{i=1}^n \|\BPhi_i - \BQ\|_F^2 } \\
& = 2\|\BPhi - \BZ\BQ\|_F \\
& \lesssim \sigma d
\end{align*}
where the third equality follows from Lemma~\ref{lem:con}.
Note that if
\[
\sigma \leq c_0 \eps \sqrt{n/d},
\]
for some small constant $c_0,$ then $d_F(\BS^0,\BZ) \leq \eps\sqrt{nd}$ which gives~\eqref{eq:kappa1-init}.

Also note that Lemma~\ref{lem:dk2} implies
\[
\|\BPhi - \BPhi^{(m)}\BQ_m\|_F \lesssim \sigma n^{-1/2}d^{1/2}(\sqrt{d} + \sqrt{\log n})
\]
where $\BQ_m = \PP((\BPhi^{(m)})^{\top}\BPhi)$ and  $|\sigma_j(\BPhi_i) - 1|\lesssim \sigma n^{-1/2}(\sqrt{d} + \sqrt{\log n}) < 1/2$ for all $1\leq j\leq d$ and $1\leq i\leq n.$ As a result, we have~\eqref{eq:kappa2-init} which follows from 
\begin{align*}
\|\BS^0 - \BS^{0,m}\BQ_m\|_F & = \|\PP_n(\BPhi) - \PP_n(\BPhi^{(m)}\BQ_m)\|_F \\
 & \leq 8\|\BPhi - \BPhi^{(m)}\BQ_m\|_F \\
& \lesssim \sigma n^{-1/2}d^{1/2}(\sqrt{d} + \sqrt{\log n}) \\
& \leq \frac{\kappa \sqrt{d}}{2}
\end{align*}
provided that 
\[
\sigma < c_0 \frac{\kappa\sqrt{n}}{\sqrt{d} + \sqrt{\log n}}
\] for some small constant $c_0.$ 

Then for the correlation between $\BW_m$ and $\BS^0 $, we have
\begin{align*}
\|\BW_m^{\top}\BS^0\|_F & \leq \|\BW_m^{\top}(\BS^0 - \BS^{0,m}\BQ_m)\|_F + \|\BW_m^{\top}\BS^{0,m}\|_F \\
& \leq \|\BW_m\|\cdot \|\BS^0 - \BS^{0,m}\BQ_m\|_F + \|\BW_m^{\top}\BS^{0,m}\|_F\\
& \leq 3\sqrt{nd}\cdot \frac{\kappa\sqrt{d}}{2} + \sqrt{nd}(\sqrt{d} + \gamma\sqrt{\log n})
\end{align*}
with probability at least $1 - n^{-\gamma^2/2}$ following from Lemma~\ref{lem:chisq} 
where all entries of $n^{-1/2}\BW_m^{\top}\BS^{0,m}\in\RR^{d\times d}$ are i.i.d. standard normal random variables and $n^{-1}\|\BW_m^{\top}\BS^{0,m}\|_F^2\sim\chi^2_{d^2}.$ Now taking the union bound over $1\leq m\leq n$ gives
\begin{align*}
\|\BW_m^{\top}\BS^0\|_F & \leq \left(1+\frac{3\kappa}{2}\right) \sqrt{nd}(\sqrt{d} + \gamma\sqrt{\log n})
\end{align*}
holds uniformly for all $1\leq m\leq n$ with probability at least $1 - n^{-\gamma^2/2+1}$ with $\gamma=4.$
\end{proof}

%\bibliography{SpectraBlock.bib}
%\bibliographystyle{abbrv}

\end{document}